\documentclass[12pt]{article}

\pdfoutput=1

\usepackage[margin=1in]{geometry}

\usepackage[
  bookmarks=true,
  bookmarksnumbered=true,
  bookmarksopen=true,
  pdfborder={0 0 0},
  breaklinks=true,
  colorlinks=true,
  linkcolor=black,
  citecolor=black,
  filecolor=black,
  urlcolor=black,
]{hyperref}
\usepackage[round]{natbib}
\usepackage{amsmath,amsthm,amssymb}
\usepackage{graphicx}

\usepackage{nicefrac}
\usepackage[capitalize]{cleveref}
\usepackage{verbatim}

\theoremstyle{plain}
\newtheorem{theorem}{Theorem}[section]
\newtheorem{lemma}[theorem]{Lemma}

\newtheorem{example}[theorem]{Example}
\newtheorem{corollary}[theorem]{Corollary}
\newtheorem{claim}[theorem]{Claim}
\Crefname{claim}{Claim}{Claims}
\newtheorem{observation}[theorem]{Observation}
\Crefname{observation}{Observation}{Observations}

\theoremstyle{definition}
\newtheorem{definition}{Definition}

\newcommand{\crefpart}[2]{\cref{#1}(\labelcref{#1-#2})}
\Crefname{equation}{Equation}{Equations}
\Crefname{figure}{Figure}{Figures}

\DeclareMathOperator{\supp}{Support}

\newcommand{\Bs}{\mathbf{s}}
\newcommand{\Bb}{\mathbf{b}}

\DeclareMathOperator*{\E}{\mathbb{E}}

\hypersetup{
  pdfauthor      = {Moshe Babaioff <moshe@microsoft.com>, Yang Cai <cai@cs.mcgill.ca>, Yannai A. Gonczarowski <yannai@gonch.name>, and Mingfei Zhao <mingfei.zhao@mail.mcgill.ca>},
  pdftitle       = {The Best of Both Worlds: Asymptotically Efficient Mechanisms with a Guarantee on the Expected Gains-From-Trade},
}

\begin{document}
\title{The Best of Both Worlds:\texorpdfstring{\\}{ }Asymptotically Efficient Mechanisms with\texorpdfstring{\\}{ }a Guarantee on the Expected Gains-From-Trade}
\author{Moshe Babaioff\thanks{Microsoft Research, \emph{E-mail}: \href{mailto:moshe@microsoft.com}{moshe@microsoft.com}.}~~~~~Yang Cai\thanks{School of Computer Science, McGill University, Canada. \emph{E-mail}: \href{mailto:cai@cs.mcgill.ca}{cai@cs.mcgill.ca}.}~~~~~Yannai A. Gonczarowski\thanks{Einstein Institute of Mathematics, Rachel \& Selim Benin School of Computer Science \& Engineering, and Federmann Center for the Study of Rationality, The Hebrew University of Jerusalem, Israel; and Microsoft Research. \emph{E-mail}: \href{mailto:yannai@gonch.name}{yannai@gonch.name}.}~~~~~Mingfei Zhao\thanks{School of Computer Science, McGill University, Canada. \emph{E-mail}: \href{mailto:mingfei.zhao@mail.mcgill.ca}{mingfei.zhao@mail.mcgill.ca}.}}
\date{February 22, 2018}

\maketitle

\begin{abstract}
The seminal impossibility result of \citet{MyersonS83} states that for bilateral trade, there is no mechanism that is individually rational (IR), incentive compatible (IC), weakly budget balanced, and efficient. This has led follow-up work on two-sided trade settings to weaken the efficiency requirement and consider approximately efficient simple mechanisms, while still demanding the other properties. The current state-of-the-art of such mechanisms for two-sided markets can be categorized as giving one (but not both) of the following two types of approximation guarantees on the \emph{gains from trade}: a constant \emph{ex-ante} guarantee, measured with respect to the \emph{second-best} efficiency benchmark, or an asymptotically optimal \emph{ex-post} guarantee, measured with respect to the \emph{first-best} efficiency benchmark. Here the second-best efficiency benchmark refers to the highest gains from trade attainable by any IR, IC and weakly budget balanced mechanism, while the first-best efficiency benchmark refers to the maximum gains from trade (attainable by the VCG mechanism, which is not weakly budget balanced).

In this paper, we construct simple mechanisms for double-auction and matching markets that \emph{simultaneously} achieve both types of guarantees: these are ex-post IR, Bayesian IC, and ex-post weakly budget balanced mechanisms that\ \ 1)~ex-ante guarantee a constant fraction of the gains from trade of the second-best, and\ \ 2)~ex-post guarantee a realization-dependent fraction of the gains from trade of the first-best, such that this realization-dependent fraction converges to $1$ (full efficiency) as the market grows large.
\end{abstract}

\section{Introduction}

In a two-sided trade setting, some agents (sellers) are endowed with items, while other agents (buyers) are interested in purchasing items. Each seller has a cost for parting with her item, and each buyer has a value for obtaining an item. In such settings, a mechanism designer may wish to create a mechanism that ensures that the items end up belonging to the agents (whether buyers or sellers) that value them the most. An important property of such a mechanism is being budget balanced, that is, not running a deficit for the mechanism designer.

The seminal impossibility result of
\citet{MyersonS83} shows that for bilateral-trade, that is, for the setting where a single seller wishes to sell a single item to a single buyer, there is no mechanism that is individually rational (IR), incentive compatible (IC), weakly budget balanced (BB) and efficient (i.e., maximizes welfare).\footnote{In particular, note that the VCG mechanism, while being IR, dominant-strategy incentive compatible, and efficient, has a budget deficit.} This impossibility result clearly extends from the special case of bilateral trade to any two-sided trade setting.

In light of the above impossibility result, follow-up work in the two-sided trade literature has looked at IR, IC, and BB mechanisms that are approximately efficient, rather than precisely efficient.
Two definitions of approximate efficiency have emerged: on the one hand, approximately maximizing welfare\footnote{These papers consider the cost of a seller as a value for keeping the item rather than a cost for parting with the item, so the no-trade welfare is the cost (or rather value) of the seller, while the post-trade welfare (if trade occurs) is the value of the buyer.} \citep{BlumrosenD16,Colini-Baldeschi16,Colini-Baldeschi17}, and on the other hand, approximately maximizing the gains from trade (GFT), that is, the increase in total welfare due to the trade \citep{McAfee92,BabaioffNP09,BlumrosenM16, DuttingTR17, BrustleCFM17}.
This paper discusses the latter, more challenging benchmark.\footnote{While maximizing the gains from trade coincides with maximizing welfare, obtaining a constant approximation to the optimal gains from trade is considerably more demanding than obtaining a constant approximation to the optimal welfare. Consider, for instance, a buyer who values an item by $9$ dollars and a seller whose value for keeping the item is $8$ dollars. The optimal welfare (of $9$), and the optimal gains from trade (of $1$), are both obtained by having the seller trade with the buyer. While not trading results in a welfare of $8$ (a \nicefrac{8}{9} fraction of the optimal welfare), it results in zero gains from trade.}
The current state-of-the-art mechanisms in the literature can be categorized as giving one of two guarantees:
\begin{enumerate}
\item A constant ex-ante guarantee, measured with respect to the ``second-best'' efficiency benchmark, that is, the (possibly very complex) mechanism obtaining the highest expected GFT of any IR and IC mechanism that is weakly budget balanced,\quad or
\item An asymptotically optimal ex-post guarantee, measured with respect to the ``first-best'' efficiency benchmark, that is, the mechanism obtaining full efficiency (VCG).
\end{enumerate}
In this paper, we aim to construct simple mechanisms that simultaneously achieve both guarantees.
We study settings in which each seller is endowed with precisely one item, all items are identical, and each buyer is interested in buying one item. In the \emph{double-auction} setting, any seller can trade with any buyer, while in the more general \emph{matching} setting, trade between some buyer-seller pairs is disallowed.
Before describing our results, we first survey the state-of-the-art mechanisms giving each guarantee in more depth.

\paragraph{Ex-Ante Guarantees}
\citet*{BrustleCFM17} (henceforth BCWZ) present a simple mechanism that is IR, IC, and weakly BB, and obtains, in expectation, at least half of the expected GFT of the (possibly very complex) second-best efficiency benchmark. More specifically, they have proposed two mechanisms -- a buyer offering and a seller offering mechanism -- and have showed that the total GFT of these two mechanisms is at least the GFT of the second-best mechanism, implying that a random one obtains at least half the GFT of the second-best mechanism.
For bilateral trade, in their seller offering mechanism, the seller simply posts a take-it-or-leave-it price to the buyer, which maximizes the seller's utility in expectation, taking into account the seller's cost for the item and the buyer's value distribution. In their buyer offering mechanism, the buyer makes a similar take-it-or-leave-it purchase offer to the seller.

BCWZ also generalize their results beyond bilateral-trade settings, to more complex two-sided trade scenarios including double-auctions and matching settings. Their mechanism for these settings generalizes the seller-offering mechanism by maximizing the total Myerson virtual surplus of the sellers for the given buyers' distributions, and similarly for the buyer-offering mechanism.
While the mechanism that they present obtains at least half of the second-best GFT in expectation,
we observe that it does not give any ex-post efficiency guarantees, and moreover, even its expected
GFT does \emph{not} asymptotically converge to the GFT of the second-best (let alone the first-best) mechanism as the market grows large.
This holds even for the very simple double-auction market
with $n$~sellers, each selling an identical item, and $n$ buyers, each interested in buying a single item, with the values (or costs) of the agents sampled i.i.d.\ from the uniform distribution over $[0,1]$. Even when $n$ is large, the mechanism of
BCWZ will only give in expectation a constant fraction (strictly smaller than $1$) of the second-best GFT, and no more than that (see Example \ref{example:rvwm-not-eff} in Section \ref{sec:naive-issues}). In particular, even in a large market, the efficiency of their mechanism does not converge to full efficiency.

\paragraph{Ex-Post Guarantees}
The Trade Reduction mechanism\footnote{\citeauthor{McAfee92}'s original mechanism is slightly more involved. We use a simplified version that provides the same worst-case guarantees.}
 of \citet{McAfee92}, which is defined for the double-auctions setting, does not suffer from the above drawback and is asymptotically efficient.
The mechanism circumvents the impossibility result of \citet{MyersonS83} for bilateral trade, by providing an ex-post efficiency guarantee only when more than one trade is possible in the double-auction market.
The mechanism works as follows: it first finds the efficient trade --- denote the size (number of pairs) of this trade by $q$. It then removes the least efficient trade (one buyer-seller pair), and only allows for the remaining trades (the $q\!-\!1$ most efficient trades) to realize, charging the winning buyers the value of the removed buyer, and paying the winning sellers the cost of the removed seller.
 This creates an IR and ex-post incentive compatible (IC) mechanism.
 As the value of the removed buyer is at least the cost of the removed seller, each trade is weakly budget-balanced. The mechanism obtains at least a $1\!-\!\nicefrac{1}{q}$ fraction\footnote{Recall that $q$ is a function of the valuation profile.} of the realized optimal (first-best) GFT. In the double-auction example above, as $n$ grows $q$ will also grow, and this fraction will tend to~$1$. Unfortunately, when $q=1$ this mechanism performs no trade and provides no guarantees at all.
(Failing to provide an ex-post guarantee unconditionally is of course inevitable in light of the impossibility result of
\citet{MyersonS83}.\footnote{Ex-post approximation to the GFT requires the mechanism to trade whenever there is positive gain, but the impossibility result implies that for some of these profiles trade will not occur.}) We note that the Trade Reduction mechanism, while asymptotically efficient, fails to give any unconditional approximation to the GFT, even with respect to the GFT of the second-best mechanism (as the mechanism of
BCWZ does give).

\paragraph{The Best of Both Worlds} In this work we aim to design simple mechanisms that are IR, IC, and weakly BB, and simultaneously provide both types of efficiency guarantees discussed above. First, in the spirit of the guarantee of
BCWZ, we aim to guarantee for the expected GFT to be at least a constant fraction of the expected GFT of the second-best mechanism. Second, in the spirit of the guarantee of \citet{McAfee92}, we aim to guarantee for the ex-post GFT to be at least a realization-dependent fraction of the realized optimal GFT (first-best), such that this fraction tends to $1$ ``as the market grows large'' and the efficient trade size grows\footnote{The condition on the efficient trade size ensures that the growth in the market size does not, for example, come from adding agents that are ``irrelevant,'' such as buyers with $0$ value and sellers with very high costs, since in such a case it would not be possible to provide any guarantee that is better for large markets than for small ones (such as bilateral trade markets).}
to infinity.

\subsection{Our Results}

We present results both for the double-auction setting and for the more involved matching-market setting, which extends the double-auction setting by adding constraints on the pairs of agents who can trade with each other.
Providing a result for this more involved scenario is considerably more challenging than for the double-auction setting, and is the main result of this paper.

\subsubsection{Double Auctions}

We first consider the double-auction setting in which each seller has a single item, all items are identical, and each buyer desires a single item. Each value (for obtaining an item) of each buyer, and each cost (for parting with her item) of each seller is drawn from a known agent-specific distribution, independently from all other values and costs.
We first present our result for double auctions.

\begin{theorem}[See \cref{thm:double-auction}]
	For the double-auction setting, there exists a simple mechanism that is ex-post IR, Bayesian IC and ex-post weakly budget balanced, and satisfies both of the following.
\begin{itemize}
\item
The expected GFT of this mechanism is at least $\nicefrac{1}{4}$ of the expected GFT of the second-best mechanism.
\item
This mechanism guarantees at least $1\!-\!\nicefrac{1}{q}$ of the realized optimal (first-best) GFT, where $q$ is the size of the most efficient trade. Thus the mechanisms is asymptotically efficient (converges to full efficiency as the trade size $q$ grows large).
\end{itemize}
\end{theorem}
Note that the asymptotic efficiency that is obtained is with respect to the most demanding benchmark of the realized optimal GFT (the first-best and not only the second-best), providing the same guarantee as the one provided by the Trade Reduction mechanism of \citet{McAfee92}. The concurrent ex-ante guarantee is with respect to the second-best, similarly to the result of
BCWZ.

Before examining the problem thoroughly, one might be tempted to think that it is trivial to come up with such a mechanism for double auctions.
Here is a natural na{\"i}ve candidate for such a mechanism: first, the mechanism  computes the efficient trade size $q$. If $q>1$, it runs McAfee's Trade Reduction mechanism. Otherwise, it runs the mechanism of
BCWZ.
This na{\"i}ve approach turns out to fail miserably as the allocation is not even monotone: it may well be the case that the two agents that trade when $q=1$ (i.e., those that trade according to the mechanism of
BCWZ) are not the highest-value buyer and the lowest-cost seller,
and so in certain scenarios an agent that is reduced in the $q>1$ case (by the Trade Reduction mechanism) may be able to reduce her bid to move to the $q=1$ case and trade (for more details see Section \ref{sec:naive-issues}).

To present our mechanism, let us first very roughly review the behavior of the mechanism of
BCWZ in the bilateral-trade case: in this special case, their mechanism flips a coin; with probability $50\%$,
the seller offers a take-it-or-leave-it price to the buyer (calculated so as to maximize the expected utility of the seller), and with probability $50\%$, the buyer offers a take-it-or-leave-it price to the seller (calculated so as to maximize the expected utility of the buyer). In order to obtain the mechanism that we seek, we {carefully} make two main modifications to the na{\"i}ve ``compound'' mechanism described above: first, in order to address the above-discussed source of non-monotonicity, instead of running the mechanism of
BCWZ on the entire market, we run their bilateral-trade mechanism only on the (unique) pair in the efficient trade. To make the resulting mechanism truthful, we need to make an additional adjustment: in the seller-offering case (the adjustment to the buyer-offering case is analogous), we on the one hand force the seller to
set a price that is
at most the threshold bid that puts her in the efficient (first-best) trade, and on the other hand notify her of the values of all buyers except the one that she is facing, and calculate the price that she offers to maximize her expected utility conditioned upon the fact that the buyer that she is facing has value larger than all of these values. Both adjustments, and in particular the first one,  make the proof of the ex-ante guarantee, as well as the proof that the mechanism is IC, quite subtle.

The  main challenge in obtaining the approximation guarantee for the case where $q=1$ is to reconcile the fact that the pair that our mechanism attempts to trade on is determined by maximizing the realized GFT (first-best) and might not be the same as the pair that would have traded according to the mechanism of
BCWZ.
The main hurdle to obtaining the approximation guarantee for this case is therefore that for some valuation profiles, an offer between the unique pair in the efficient trade will be rejected, resulting in no trade in our hybrid mechanism, while in the mechanism of
BCWZ an offer will be made --- and accepted --- between a different pair. To overcome this, we have to carefully charge such losses in GFT to gains in GFT by other parts of our hybrid mechanism.

\subsubsection{Matching Markets}

As stated above, the mechanism of
BCWZ does not converge to the efficient outcome in large double-auction markets, and thus will clearly not do so in the more general matching market setting. Our goal is to present a mechanism for matching market that is IR, IC and ex-post weakly BB, but also provide ex-ante guarantees for the GFT as well as ex-post guarantees that converges to full efficiency ``as the market grows large''.
While in the double auction setting, every buyer can trade with every seller, this is no longer the case in a matching market. Our notion of a large matching market aims to generalize the fact that in a large market there are many agents that are ``equivalent'' up to their values. The sense of agents being equivalent in a matching market is that they can trade with exactly the same set of agents. So, we can naturally partition agents to equivalence classes, with every two agents of the same class being interchangeable in any matching (up to their valuations). We consider matching markets with a fixed set of such classes, and think about a large market as a market in which the number of agents of each class is growing large, yet the number of different classes that any agent can trade with remains bounded by some constant $d$.

Recall that the Trade Reduction mechanism of \citet{McAfee92} is defined for a double-auction setting.
We first present a generalization for matching markets of the Trade Reduction mechanism (\cref{sec:matching-tr}) and prove that it is ex-post asymptotically efficient ``as the market grows large'' in the above sense. To our knowledge, this nontrivial generalization of the Trade Reduction mechanism, which may also be of separate interest, is novel. Similarly to the Trade Reduction mechanism of \citet{McAfee92} for double-auction settings, this mechanism does not give any ex-ante approximation guarantee.

As with the double-auction case, we cannot directly combine our Trade Reduction mechanism for matching markets with the mechanism of
BCWZ while maintaining truthfulness. Therefore, we present a novel mechanism (\cref{sec:matching-offering}), which we call the \emph{Offering Mechanism for Matching Markets}.
Like the mechanism of
BCWZ, this mechanism does not provide the ex-post guarantee we are after, but we manage to carefully define it in a way that allows us to combine it with the Trade Reduction mechanism for matching markets to obtain {a truthful mechanism that} provides both types of guarantees that we are after.
{The precise definition of the Offering Mechanism that allows for both the truthfulness and the efficiency guarantees of the hybrid mechanism has been quite elusive to pin down, and the proofs of truthfulness, and in particular of the ex-ante guarantee, are considerably more subtle than in the double-auction setting.}
To prove the ex-ante guarantee of the Offering Mechanism, we compare it to the mechanism of
BCWZ, showing that it attains at least half of the GFT of their mechanism, resulting in an ex-ante guarantee of at least $\nicefrac{1}{4}$ of the expected GFT of the second best mechanism. Proving the ex-ante guarantee of the Offering Mechanism is the most technically challenging part of our analysis.
To prove this guarantee, we show that it is possible to carefully ``charge'' every edge of the matching of
BCWZ to edges of the first-best matching that will be traded in our Offering Mechanism, proving that
the expected GFT of our Offering Mechanism is at least half the expected GFT of the mechanism of
BCWZ.
The combination of the Offering Mechanism for matching markets with the Trade Reduction mechanism for matching markets creates the \emph{Hybrid Mechanism for Matching Markets} (\cref{sec:matching-hybrid}), giving us our main result.

\begin{theorem}[See \cref{thm:matching-main}]
	For the matching market setting, there exists a simple mechanism that is ex-post IR, Bayesian IC and ex-post weakly budget balanced, and satisfies both of the following.
	\begin{itemize}
		\item
The expected GFT of this mechanism is at least $\nicefrac{1}{4}$ of the expected GFT of the second-best mechanism.
		\item When $1\!-\!\nicefrac{d}{q}\geq \nicefrac{1}{2}$, this mechanism guarantees at least $1\!-\!\nicefrac{d}{q}$ of the realized optimal \mbox{(first-best)} GFT where $d$ denotes the maximum number of classes that
		any agent can trade with, and $q$ denotes the minimal positive number of trading agents of the same class in the welfare maximizing outcome.
		Thus the mechanism is asymptotically efficient in the sense that it converges to full efficiency as the number of trading agents in every class grows large, as long as $d$ is fixed.
	\end{itemize}
\end{theorem}

We remark that while our mechanism ex-ante guarantees a quantitatively smaller fraction of the second-best GFT than the $\nicefrac{1}{2}$ fraction guaranteed by the mechanism of
BCWZ, our mechanism has two qualitative advantages over their mechanism: first, we additionally obtain an ex-post guarantee on the GFT that is asymptotically efficient; and moreover, while both mechanisms ensure that a truthful agent \emph{never} regrets participating (ex-post IR), our mechanism is guaranteed to \emph{never} lose money, while theirs gives this guarantee only in expectation and sometimes runs a deficit.
	
\subsection{Additional Related Work}

The Trade Reduction Mechanism of \citet{McAfee92} was generalized to different settings to provide similar asymptotic efficiency guarantees as well as ex-post guarantees as a function of the trade size with IR, IC mechanisms that are budget balanced. \citet{BabaioffW05} have presented Trade Reduction mechanisms for Supply Chain settings, while \citet{BabaioffNP09} presented such a mechanism for a Spatially Distributed Market. In Section \ref{sec:matching-tr} we generalize the Trade Reduction mechanism to matching markets.

Recent papers \citep{BlumrosenD16, BlumrosenM16} have focused on IR and Bayesian IC mechanisms that guarantee approximate efficiency while maintaining budget balance.
\citet{BlumrosenD16} have presented a mechanism for bilateral-trade that is strongly budget balanced and obtains in expectation at least a constant fraction of the optimal welfare (the optimal welfare is the higher of the values of the two agents for the item). \citet{BlumrosenM16} have considered the more challenging goal of approximately maximizing the GFT, and have presented a mechanism that obtains in expectation at least $\nicefrac{1}{e}$ of the first-best GFT when the buyer's valuation is drawn from a distribution satisfying the monotone hazard rate condition.
\citet{DuttingTR17} have studied the prior-free setting and have designed ex-post IC mechanisms that approximate the GFT and are budget-balanced for two-sided markets with constraints on each side separately, but leave open the design problem when there are cross-market constraints, which we study in our paper. Recently, \citet{Colini-Baldeschi16} have showed how to design an IR, ex-post IC and strongly BB mechanism for the double auction setting where there may be matroid feasibility constraint on the set of buyers who can trade simultaneously. Their mechanism achieves a constant fraction of the ex-ante optimal social welfare, but provides no guarantee on the GFT. Moreover, their mechanism is not asymptotically efficient even when the market grows large. Finally, \citet{Colini-Baldeschi17} have considered a two-sided combinatorial auctions, where the market has multiple types of items for sale. Each seller might own multiple items
and she has additive valuation over her items. Every buyer has XOS valuation over the items. They have showed that a variant of a sequential posted price mechanism can achieve a constant fraction of the optimal social welfare. Their mechanism neither provides any ex-post guarantees nor converges to efficiency in large market. Indeed, their mechanism may not trade a single pair of buyer and seller even when there are many tradeable pairs~\footnote{This could happen when every item's expected contribution to the social welfare is not much bigger than its expected cost.}.

\section{Preliminaries}
\label{sec:preliminaries}

\subsection{Model and Definitions}

\paragraph{Agents and Utilities}
In a market for identical goods, there is a finite set $S$ of sellers with one good each, and a finite set $B$ of unit-demand buyers, with $|S|\ge2$ and $|B|\ge2$. Each seller $j\in S$ has a cost $s_j>0$ that she incurs if she sells her item, and each buyer $i\in B$ has a value $b_i>0$ that she derives if she purchases an item. We assume that an agent who does not trade does not incur any cost or derive utility from this. Let $\Bs$ be the vector of sellers' costs and $\Bb$ be the vector of buyers' values.
The costs and values are sampled from agent-specific (but not necessarily identical) nonnegative distributions $D_i^B$ for each buyer $i\in B$ and $D_j^S$ for each seller $j\in S$, each independent of all other distributions. Agents have quasi-linear utilities and are risk neutral.

\paragraph{Trading Constraints}
In a matching market setting, an undirected bipartite graph $G=(S,B,E)$ with the sellers on one side and the buyers on the other constraints transactions. A set of \emph{trading agents} $K$ is a set of buyers and of sellers
that can be partitioned into pairs, each of one buyer and one seller that are neighbors in $G$ (this is equivalent to a matching of the set $K$ in $G$)
---
the set $K$ corresponds to each seller selling her item, and each buyer buying one of the items sold from one of its neighbors in $G$.
The \emph{size of trade} of $K$ is defined to be $|K\cap S|=|K\cap B|$.

\paragraph{Gains from Trade}
The \emph{gains from trade (GFT)} when the set $K$ (of trading agents) is trading is defined to be $\sum_{i\in K\cap B} b_i - \sum_{j\in K\cap S} s_j$.
Given a valuation profile $(\Bb,\Bs)$, a set of trading agents is \emph{efficient}
if it maximizes the gains from trade from $(\Bb,\Bs)$ among all sets of trading agents.

\paragraph{Mechanisms}
We consider direct-revelation mechanisms in which each agent reports a \emph{type} (value for buyers, cost for sellers), so the mechanism is a function from reported valuation profiles to a set of trading agents and to payments from each agent to the mechanism.
A mechanism is \emph{Bayesian incentive compatible (BIC)} if each agent, by being truthful, maximizes her expected utility (over the randomization of the mechanism and the types of the others, assuming that they are truthful\footnote{Our BIC mechanisms will actually satisfy a slightly stronger truthfulness property, being truthful for every realization of the coins of the mechanism, yet only in expectation over the types of the other agents.}.) A mechanism is \emph{ex-post IC} if being truthful maximizes the agent's utility for any actions (reports) of the others.
A mechanism is \emph{(ex-post) IR} if the utility for a truthful agent is non-negative, independent of the strategies of others.
A mechanism is \emph{interim IR} if the expected utility for a truthful agent is non-negative, when the expectation is over the randomization of the mechanism and the types of the other agents, when truthful. Clearly, if a mechanism is ex-post IR, then it is also interim IR.
As all the mechanisms in this paper are ex-post IR and BIC (or even ex-post IC), then unless otherwise stated, we assume that the reports are equal to the true values/costs.

A mechanism is \emph{ex-post weakly budget balanced (BB)} if for any valuation profile, the sum of payments to the mechanism is non-negative. A mechanism is \emph{ex-post strongly budget balanced} if for any valuation profile, the sum of total payments to the mechanism is zero.
A truthful mechanism is \emph{ex-ante weakly budget balanced} if the expected sum of payments to the mechanism is non-negative, where the expectation is over the types of all agents and the randomness of the mechanism. Clearly, if a mechanism is ex-post weakly budget balanced, then it is also ex-ante weakly budget balanced.
Following \citet{Colini-Baldeschi17} we say that one of the above budget balance properties holds for \emph{direct trade} if that budget balance property (weak or strong) holds for each of the trades separately.

\paragraph{Benchmarks}
Given a valuation profile $(\Bb,\Bs)$, let $M(\Bb,\Bs)$ be the \emph{first-best} matching, or the maximum-weight matching in $G$, where ties between agents are broken by the ``lexicographic order by IDs'' formally defined in Definition \ref{def:lex-order} in \cref{app:ties}.\footnote{This tie breaking rule satisfies two properties we use extensively: 1) it does not depend on weights, and 2) it is \emph{subset consistent} in the sense that when removing an edge $(i,j)$ from some matching $M$ and picking a matching on the remaining nodes
	$M\setminus \{i,j\}$, it will pick the matching of $M$ on these nodes.}
Slightly abusing notation, we use $M(\Bb,\Bs)$ to also denote the set of agents in the matching $M(\Bb,\Bs)$.
Let $OPT(\Bb,\Bs)$ be the GFT of the ``first-best'' $M(\Bb,\Bs)$, that is $OPT(\Bb,\Bs)=\sum_{(i,j)\in M(\Bb,\Bs)} (b_i-s_j)$.
Note that the VCG mechanism (which is not budget balanced) attains a GFT of $OPT(\Bb,\Bs)$ on every valuation profile $(\Bb,\Bs)$.

A mechanism is called \emph{second-best} if it maximizes the expected gains from trade among all BIC, interim IR and ex-ante weakly budget balanced mechanisms.

\paragraph{Special Cases}
The case where $G$ is the \emph{complete} bipartite graph (i.e., any seller can trade with any buyer) is called the \emph{double-auction} setting.
In the double-auction setting, for every valuation profile $(\Bb,\Bs)$ we denote the size of the efficient set of trading agents by $q(\Bb,\Bs)$.
The case where $|S|=|B|=1$ and the buyer and the seller are connected by an edge in $G$ (so this is also a special case of double-auction) is called the \emph{bilateral-trade} setting.

\subsection{The Trade Reduction Mechanism}

In the double-auction setting, the \emph{Trade Reduction (TR) mechanism} \citep{McAfee92} is a mechanism that finds the most efficient trade of only $q(\Bb,\Bs)-1$ items,\footnote{If $q(\Bb,\Bs)=0$ there is no trade in the TR mechanism, and no payments are made.}
and charges each agent his critical value for winning.
That is, the $q(\Bb,\Bs)\!-\!1$ highest-value buyers trade and pay the bid of the reduced buyer (the $q(\Bb,\Bs)$-highest buyer); they trade with the $q(\Bb,\Bs)\!-\!1$ lowest-cost sellers, each seller getting paid the cost of the reduced seller (the $q(\Bb,\Bs)$-lowest seller).

\begin{theorem}[\citealp{McAfee92}]\label{tr}
In the double-auction setting, the TR mechanism is ex-post IC, ex-post IR, and ex-post (direct trade) weakly budget balanced. For every valuation profile $(\Bb,\Bs)$, the gains from trade of this mechanism are at least an $1\!-\!\frac{1}{q(\Bb,\Bs)}$ fraction of $OPT(\Bb,\Bs)$.
\end{theorem}

Note that if $q(\Bb,\Bs)=1$, then no ex-ante approximation to the GFT is achieved by the TR mechanism,
while for $q(\Bb,\Bs)\geq 2$, \cref{tr} guarantees at least half the efficient GFT for $(\Bb,\Bs)$, ex-post.

\subsection{The Random Virtual-Welfare Maximizing Mechanism of\texorpdfstring{\\}{ }\texorpdfstring{\citet{BrustleCFM17}}{Brustle et al. [2017]}}\label{sec:preliminaries-rvwm}

\citet{BrustleCFM17} present a mechanism for trading with downward-closed constraints (which subsume matching constraints), which we will refer to throughout this paper as the \emph{Random Virtual-Welfare Maximizing (RVWM)} mechanism.
We will first describe this mechanism, and then distill from this description the information that will be required for our analysis.
The mechanism is described in terms of the ironed virtual value and virtual cost functions of the agents. For any buyer $i$, the \emph{ironed virtual value function}\footnote{When the CDF of $D^B_i$ is differentiable, then the (non-ironed) virtual value $\varphi_i(b_i)$ of seller $i$ with value $b_i$ is defined as $b_i-\frac{1-D^B_i(b_i)}{d^B_i(b_i)}$, where $D^B_i$ and $d^B_i$ are the CDF and PDF of the distribution $D^B_i$ from which buyer $i$'s value is drawn. If the virtual value function $\varphi_i$ is not non-decreasing, then we perform an ironing procedure to make it monotone, resulting in the \emph{ironed} virtual value function $\tilde{\varphi}_i$. We refer the reader, e.g., to \citet{BrustleCFM17} for more details.} \citep{Myerson81} of buyer $i$ is denoted by $\tilde{\varphi}_i(\cdot)$ and for the purposes of our analysis it is enough to observe that it is a non-decreasing function such that for every value $b_i$ we have $\tilde{\varphi}_i(b_i)\le b_i$. For any seller $j$, the \emph{ironed virtual cost function}\footnote{This function is defined symmetrically to \citeauthor{Myerson81}'s ironed virtual value function. When the CDF of $D^S_j$ is differentiable, then the (non-ironed) virtual cost $\tau_j(s_j)$ of seller $j$ with cost $s_j$ is defined as $s_j+\frac{D^S_j(s_j)}{d^S_j(s_j)}$, where $D^S_j$ and $d^S_j$ are the CDF and PDF of the distribution $D^S_j$ from which seller $j$'s cost is drawn. If the virtual cost function $\tau_j$ is not non-decreasing, then we perform an ironing procedure to make it monotone, resulting in the \emph{ironed} virtual cost function $\tilde{\tau}_j$. We refer the reader to \citet{BrustleCFM17} for more details.} of seller $j$ is denoted by $\tilde{\tau}_j$ and for the purposes of our analysis it is enough to observe that it is a non-decreasing function such that for every cost $s_j$ we have $\tilde{\tau}_j(s_j)\ge s_j$.

This RVWM mechanism flips a coin to uniformly pick one of the following two mechanisms to run:
\begin{itemize}
\item Generalized Seller-Offering Mechanism (GSOM): Given valuation profile $(\Bb,\Bs)$, let $M_1^*(\Bb,\Bs)$ be the maximum weight matching\footnote{Follow the same breaking tie rules as the first-best matching.} of $G$ when the weight of every edge $(i,j)\in E$ is $\tilde{\varphi}_i(b_i)-s_j$.
For every pair $(i,j)\in M_1^*(\Bb,\Bs)$, trade buyer $i$ with seller $j$. The allocation rule is monotone and every agent pays (or receives) her critical value to trade in the mechanism.
\item Generalized Buyer-Offering Mechanism (GBOM): Given valuation profile $(\Bb,\Bs)$, let $M_2^*(\Bb,\Bs)$ be the maximum weight matching\footnote{Follow the same breaking tie rules as the first-best matching.} of $G$ when the weight for every edge $(i,j)\in E$ is $b_i-\tilde{\tau}_j(s_j)$.
For every pair $(i,j)\in M_2^*(\Bb,\Bs)$, trade buyer $i$ with seller $j$. The allocation rule is monotone and every agent pays (or receives) her critical value to trade in the mechanism.
\end{itemize}

The only additional properties of the ironed virtual value and cost functions that our analysis will require will be used through the following \lcnamecref{rvwm-offer}.

\begin{observation}\label{rvwm-offer}
Let $(\Bb,\Bs)$ be a valuation profile.
If trade occurs with some positive probability on a given edge $(i,j)$ in the RVWM mechanism, then trade would occur on the same edge with at least the same probability in the mechanism that runs one of the following, with probability $50\%$ each:
\begin{itemize}
\item
Seller $j$ offers a price to buyer $i$ that maximizes the utility of seller $j$ in expectation over the distribution $D^B_i$ from which buyer $i$'s valuation was drawn, and trade occurs if and only if this price is at most buyer $i$'s valuation $b_i$.
\item
Buyer $i$ offers a price to seller $j$ that maximizes the utility of buyer $i$ in expectation over the distribution $D^S_j$ from which seller $j$'s valuation was drawn, and trade occurs if and only if this price is at least seller $j$'s cost $s_j$.
\end{itemize}
\end{observation}

To see why \cref{rvwm-offer} follows from the above definition, note that if for a valuation profile $(\Bb,\Bs)$ there is trade with positive probability on an edge $(i,j)$, then either GSOM or GBOM traded that edge. If GSOM traded that edge, then it means that $\tilde{\varphi}_i(b_i)-s_j\ge0$. So, $\tilde{\varphi}_i^{-1}(s_j)\le b_i$. \citet{Myerson81} shows that $\tilde{\varphi}_i^{-1}(s_j)$ is a price that when offered, maximizes the expected utility of seller $j$ with cost $s_j$ from buyer $i$ (the ironed virtual value function $\tilde{\varphi}_i$ encodes the valuation distribution of buyer $i$). So, since this price is at most $b_i$, it would have been accepted in the seller-offering mechanism described in \cref{rvwm-offer}. If GBOM traded this edge, then similarly an offer would have been accepted in the buyer-offering mechanism described in \cref{rvwm-offer}.

BCWZ prove that the RVWM mechanism guarantees at least half of the GFT of the second-best ex-ante:

\begin{theorem}[\citealp{BrustleCFM17}]\label{bcfm}
The RVWM mechanism for downward-closed constraints of \citet{BrustleCFM17} is ex-post IC, ex-post IR, and ex-ante weakly budget balanced, and in expectation gets a $\nicefrac{1}{2}$-fraction of the gains from trade
of the second-best mechanism.
\end{theorem}

Note that while the ex-ante guarantee on the GFT of
the RVWM mechanism (\cref{bcfm}) is with respect to the \emph{second-best} mechanism, the ex-post guarantee on the GFT of the TR mechanism (\cref{tr}) is with respect to the more demanding benchmark of the (realized) \emph{first-best} gains from trade. Also note that while the RVWM mechanism is ex-post IC like the TR mechanism, it is only ex-ante, rather than ex-post, weakly budget balanced. While our main result will be stated to guarantee that our mechanism is BIC and ex-post weakly budget balanced, we will note that our mechanism can also be made ex-post IC for the price of being only ex-ante weakly budget balanced, thus matching these guarantees of the RVWM mechanism (while adding an asymptotically efficient ex-post guarantee on the GFT).

\section{Shortcomings of the RVWM Mechanism and of\texorpdfstring{\\}{ }Na{\"i}ve Modifications thereto}
\label{sec:naive-issues}

In this \lcnamecref{sec:naive-issues}, we demonstrate the shortcomings of the RVWM mechanism that motivate our work, as well as the ineffectiveness of na{\"i}ve modifications to this mechanism in overcoming these shortcomings. We first show that the RVWM mechanism of \citet{BrustleCFM17} is not asymptotically efficient, even ex-ante, and then show that two na{\"i}ve strategies to combine this mechanism with the Trade Reduction mechanism of \citet{McAfee92} are not monotone, even interim, and therefore there is no hope in coupling them with a payment rule so as to make them Bayesian IC.

\subsection{Asymptotic Inefficiency of the RVWM Mechanism}

We first observe that the the RVWM mechanism of \citet{BrustleCFM17} is not asymptotically efficient for double auctions, even ex-ante and compared to the second-best.

\begin{example}\label[example]{example:rvwm-not-eff}
Consider a double-auction market with $n$ seller and $n$ buyers, with agents' values and costs sampled i.i.d.\ from the uniform distribution over $[0,1]$. We claim that even when $n$ is large, the RVWM mechanism will only give in expectation a  constant fraction (strictly smaller than $1$) of the expected GFT of the second-best
mechanism. In particular, even in a large market, and even in expectation, the efficiency of the RVWM mechanism with respect to the second best (and thus also with respect to the first-best) does not converge to full efficiency.
\end{example}

\begin{proof}[Proof sketch]
We prove \cref{example:rvwm-not-eff} in \cref{app:examples}, and here we give some intuition. When $n$ is large, it is easy to observe that
in an efficient trade roughly the $\nicefrac{n}{2}$ lowest-cost sellers (essentially distributed uniformly in $[0,\nicefrac{1}{2}]$) will sell their items to roughly the $\nicefrac{n}{2}$ highest-value buyers (essentially distributed uniformly in $[\nicefrac{1}{2},1]$),
increasing the welfare by about $\nicefrac{1}{2}$ in expectation in each trade, resulting in the first-best having asymptotic expected GFT of about $\nicefrac{n}{4}$. The second-best mechanism gets GFT that is in expectation at least the expected GFT of the Trade Reduction mechanism, so it has asymptotic expected GFT of about $\nicefrac{n}{4}-1$, asymptotically the same as the first-best mechanism.
On the other hand, when a buyer offers an optimized price facing a uniform distribution as in the RVWM mechanism, she offers only half of her value (and similarly, a seller offers a price that is half-way between her cost and $1$).
This results in only roughly the $\nicefrac{n}{3}$ lowest-cost sellers (essentially distributed uniformly in $[0,\nicefrac{1}{3}]$) selling their items to roughly the $\nicefrac{n}{3}$ highest-value buyers (essentially distributed uniformly in $[\nicefrac{2}{3},1]$), increasing the welfare by about $\nicefrac{2}{3}$ in expectation in each trade, resulting with asymptotic expected GFT of about $\nicefrac{2n}{9}<\nicefrac{n}{4}$.
\end{proof}

\subsection{Nonmonotonicity of Na{\"i}ve Modifications to the\texorpdfstring{\\}{ }RVWM Mechanism}

If one were not interested in incentive compatibility, then getting the ``best of both worlds'' would have been extremely simple: compute the outcome of both the RVWM and the TR mechanisms, and choose the outcome with higher realized GFT. As we now observe, this allocation rule is not monotone, even in an interim sense. Thus, there is no hope to couple this allocation rule with payments that will make it even Bayesian IC.

\begin{example}\label[example]{example:max-not-monotone}
Consider a double-auction setting with two buyers and two sellers. The value of buyer~$1$ is drawn from $U([0,90])$, the value of buyer~$2$ is drawn from $U([0,30])$, the cost of seller~$1$ is fixed to be $0$ with probability~$1$, and the cost of seller~$2$ is $0$ with probability $\nicefrac{1}{5}$ and is $25$ with probability $\nicefrac{4}{5}$. For the allocation rule that chooses the outcome with higher realized GFT among the outcomes of RVWM and TR, buyer~$2$ with valuation $b_2=24$ trades with higher probability (over the distributions of all other values and costs, and over the randomness of the mechanism) than buyer~$2$ with valuation $b_2=26$.
\end{example}

\begin{proof}[Proof sketch]
We prove \cref{example:max-not-monotone} in \cref{app:examples}, and here show only that this mechanism is not ex-post (rather than interim) monotone. So, fix $s_2=25$ and $b_1=30$ (and $s_1=0$). So, $\tilde{\tau}_1(s_1)=\tau_1(s_1)=0$, $\tilde{\tau}_2(s_2)=\tau_2(s_2)=25+\frac{25}{4}>30$, $\tilde{\varphi}_1(b_1)=\varphi_1(b_1)=2b_1-90=-30$, and $\tilde{\varphi}_2(b_2)=\varphi_2(b_2)=2b_2-30$.

Regardless of whether $b_2=26$ or $b_2=24$, since $\tilde{\tau}_2(s_2)>b_2$ and since $b_1>b_2$, GBOM chooses neither buyer~$2$ nor seller~$2$ as traders. Furthermore, since $s_2>\tilde{\varphi}_1(b_1)$ and since
$\tilde{\varphi}_2(b_2)\ge\tilde{\varphi}_2(24)=18>\tilde{\varphi}_1(b_1)$, GSOM chooses neither buyer~$1$ nor seller~$2$ as traders, and so since $\tilde{\varphi}_2(b_2)=18>0=s_1$, in GSOM buyer~$2$ and seller~$1$ trade. So, regardless of whether $b_2=26$ or $b_2=24$, buyer~$2$ trades in RVWM with probability $50\%$.

If $b_2=26$, then the first-best matches all agents and so in TR buyer~$1$ and seller~$1$ trade (so buyer~$2$ does not trade). So, in this case the GFT of TR is higher than the expected GFT (over the randomness of the mechanism) of RVWM, and so the TR outcome is chosen and buyer~$2$ does not trade. Conversely, if $b_2=24$, then the first-best matches only buyer~$1$ and seller~$2$ (since $b_2<s_2$) and so there is no trade in TR and the RVWM outcome is chosen and buyer~$2$ trades with probability $50\%$.
\end{proof}

Another na{\"i}ve way to combine the RVWM and TR mechanisms may be based on the value of $q$: if $q(\Bb,\Bs)\ge2$ (TR gives an ex-post guarantee), then choose the TR outcome, and otherwise choose the RVWM outcome. As we now observe, this allocation rule is not monotone either, even in an interim sense. Thus, there is also no hope to couple this allocation rule with payments that will make it even Bayesian IC.

\begin{example}\label[example]{example:rvwm-switch-not-monotone}
Consider a double-auction setting with two buyers and two sellers. The value of buyer~$1$ is drawn from $U([0,90])$, the value of buyer~$2$ is drawn from $U([0,30])$, the cost of seller~$1$ is fixed to be $0$ with probability~$1$, and the cost of seller~$2$ is $0$ with probability $\nicefrac{1}{5}$ and is $25$ with probability $\nicefrac{4}{5}$. For the allocation rule that chooses the TR outcome if $q(\Bb,\Bs)\ge2$ and the RVWM outcome otherwise, buyer~$2$ with valuation $b_2=24$ trades with higher probability (over the distributions of all other values and costs, and over the randomness of the mechanism) than buyer~$2$ with valuation $b_2=26$.
\end{example}

\begin{proof}[Proof sketch]
We prove \cref{example:rvwm-switch-not-monotone} in \cref{app:examples}, and here only note that the above proof that the mechanism from \cref{example:max-not-monotone} is not ex-post monotone in fact also shows that the mechanism from \cref{example:rvwm-switch-not-monotone} is not ex-post monotone, as these two mechanism coincide on the two valuation profiles used in the above proof of the lack of ex-post monotonicity.
\end{proof}

As can be seen from the analysis of both \cref{example:max-not-monotone,example:rvwm-switch-not-monotone} (even already from the proof of lack of ex-post monotonicity), a main source of the issues described in these \lcnamecrefs{example:rvwm-switch-not-monotone} is that the RVWM mechanism may choose as traders agents who are not in the first-best. Our approach in this paper will indeed offer an alternative to the RVWM mechanism that only chooses as traders agents who are in the first-best, and despite this added restriction still gives a qualitatively similar ex-ante guarantee to that of the RVWM mechanism. While in double auctions settings this alternative to RVWM can be considered modification of the RVWM mechanism (see \cref{sec:ro}), in matching settings this alternative is substantially different than the RVWM mechanism (see \cref{sec:matching-offering}).

\section{The Seller-Offering, Buyer-Offering, and\texorpdfstring{\\}{ }Randomized-Offerer Mechanisms}
\label{sec:ro}

Before we turn to our main results, in this \lcnamecref{sec:ro} we present a slightly modified version of the bilateral-trade construction of \citet{BrustleCFM17}, which we will use as a building block in the construction of our hybrid mechanisms, and prove several properties thereof.

\begin{definition}[SO, BO, RO Mechanisms]
Fix $D_s$ and $D_b$ to be nonnegative distributions,
and fix $\bar{s}\ge\sup\supp D_s$ and $\bar{b}\le\inf\supp D_b$ s.t.\ $\bar{s}\ge\bar{b}$.
We define three mechanisms for trade between a seller with cost $s\sim D_s$ and a buyer with value $b\sim D_b$.
\begin{itemize}
\item
The \emph{Seller-Offering (SO) mechanism with offer constraint $\bar{s}$ and target distribution $D_b$} is the mechanism in which a seller with cost $s$ offers to the buyer the lowest price $p$ among the prices that maximize the utility of the seller in expectation over $b\sim D_b$, under the constraint $p\le\bar{s}$. That is, the offered price is $\min\bigl\{p~\big|~p\in\arg\max_{p\le\bar{s}}(p-s)\cdot\bigl(1-D_b(p)\bigr)\bigr\}$. The buyer accepts this price if it is no greater than the realized value $b$ of the buyer. If the buyer accepts this price, then trade occurs at this price; otherwise, no trade occurs.
\item
The \emph{Buyer-Offering (BO) mechanism with offer constraint $\bar{b}$ and target distribution $D_s$} is the mechanism in which a buyer with value $b$ offers to the seller the highest price $p$ among the prices that maximize the utility of the buyer in expectation over $s\sim D_s$, under the constraint $p\ge\bar{b}$. That is, the offered price is $\max\bigl\{p~\big|~p\in\arg\max_{p\ge\bar{b}}(b-p)\cdot\bigl(1-D_s(p)\bigr)\bigr\}$. The seller accepts this price if it is no less than the realized cost $s$ of the seller. If the seller accepts this price, then trade occurs at this price; otherwise, no trade occurs.
\item
The (Bilateral) \emph{Randomized Offerer (RO) mechanism with SO parameters $\bar{s}$ and $D_b$ and BO parameters $\bar{b}$ and $D_s$} is the mechanism that flips a coin, with probability $\nicefrac{1}{2}$ it runs the SO mechanism with offer constraint $\bar{s}$ and target distribution $D_b$, and otherwise it runs the BO mechanism with offer constraint $\bar{b}$ and target distribution $D_s$.
\end{itemize}
\end{definition}

We slightly strengthen the special case of the incentive and budget guarantees of \cref{bcfm} for bilateral trade, and prove that they still hold even with offer constraints as in the RO mechanism.\footnote{We note that each of the SO and BO mechanisms is a deterministic and ex-post monotone mechanism, and so can be made ex-post IC (and ex-post IR) by charging the threshold winning prices. The resulting modified mechanisms, however, are not ex-post (even weakly) budget balanced, but only ex-ante (strongly) budget balanced.} We furthermore show that whenever trade occurs, the trading happens at a price that indeed lies between the seller's and the buyer's constraint.

\begin{sloppypar}
\begin{lemma}\label[lemma]{ro}
Fix $D_s$ and $D_b$ to be nonnegative distributions and fix $\bar{s}\ge\sup\supp D_s$ and $\bar{b}\le\inf\supp D_b$ s.t.\ $\bar{s}\ge\bar{b}\ge0$.
Consider the RO mechanism with SO parameters $\bar{s}$ and $D_b$ and BO parameters $\bar{b}$ and $D_s$.
\begin{enumerate}
\item\label[part]{ro-properties}
When valuations are drawn from $D_s\times D_b$, this mechanism is a BIC\footnote{Since the allocation rule of the RO mechanism is ex-post monotone, by charging threshold prices we could strengthen the incentive-compatibility property from BIC to ex-post IC (while maintaining ex-post~IR), but then the weak-budget-balance guarantee would only hold ex-ante and not ex-post (similarly to the guarantee of \cref{bcfm}). Moreover, once we settle for ex-ante budget balance, we could get ex-ante \emph{strong} budget balance ``for free'' by equally dividing our ex-ante expected profits (assuming truthful bidding) among the agents, \citep[see, e.g.,][]{BrustleCFM17}.}, ex-post IR, and ex-post (direct trade) strongly budget balanced mechanism.
\item\label[part]{ro-price-within-constraints}
Whenever trade occurs in this mechanism, it holds that the price $p$ that the seller pays the buyer satisfies $\bar{b}\le p\le\bar{s}$.
\end{enumerate}
\end{lemma}
\end{sloppypar}

The proof of \cref{ro} is given in \cref{app:ro}.
To conclude this section, we will prove two more properties of the RO mechanisms that will allow us to lower-bound its GFT guarantee: the first will allow us to compare its GFT to that of the first-best, and the second will allow us to compare its GFT to that of the RVWM mechanism.

\begin{lemma}\label[lemma]{ro-trade-probability}
Fix $D_s$ and $D_b$ to be nonnegative distributions and fix $\bar{s}\ge\bar{b}\ge0$.
Fix $s\le\bar{s}$ to be a cost for the seller and fix $b\ge\bar{b}$ to be a value for the buyer.
Consider the RO mechanism with SO parameters $\bar{s}$ and $D_b|_{\ge\bar{b}}$ and BO parameters $\bar{b}$ and $D_s|_{\le\bar{s}}$.\footnote{For a distribution $D$ and a value $c$, we use $D|_{\le c}$ to denote this distribution conditioned upon the drawn value being at most $c$, and use $D|_{\ge c}$ to denote this distribution conditioned upon the drawn value being at least $c$.}
\begin{enumerate}
\item\label[part]{ro-trade-probability-always-trade}
If $\bar{b}\ge s$ or $\bar{s}\le b$, then the probability that trade occurs in this mechanism is at least~$\nicefrac{1}{2}$.
\item\label[part]{ro-trade-probability-restrictions-dont-matter}
If $\bar{b}\le s$ and $\bar{s}\ge b$, then the probability that trade occurs in this mechanism is at least as high as the probability that trade occurs in the RO mechanism with SO parameters $\infty$ and $D_b$ and BO parameters $0$ and $D_s$.
\end{enumerate}
\end{lemma}

The proof of \cref{ro-trade-probability} is given in \cref{app:ro}. In a nutshell, \cref{ro-trade-probability-always-trade} holds since if, e.g., $\bar{b}\ge s$, then an offer by the buyer will always be accepted by the seller, and  \cref{ro-trade-probability-restrictions-dont-matter} holds since under the given assumptions, if trade occurs in the unconstrained and unconditioned RO mechanism, then the price offered there also satisfies all of the extra restrictions of the constrained and conditioned RO mechanism, and therefore the same price will be offered in that mechanism as well, resulting in trade there as well.

\section{Double Auctions}
\label{sec:double-auctions}

In this \lcnamecref{sec:double-auctions}, we present our results for the double-auction setting, in which there are no constraints on which seller can trade with which buyer (i.e., the graph $G$ is the full bipartite graph). Namely, we will present our hybrid mechanism for double auctions, which is an ex-post IR, BIC, ex-post weakly budget balanced mechanism, which ex-ante guarantees a constant fraction of the second-best, and is ex-post asymptotically efficient.

\subsection{A Hybrid Mechanism for Double Auctions}

While the RVWM mechanism is not asymptotically efficient, the Trade Reduction (TR) mechanism of \citet{McAfee92} is asymptotically efficient as it guarantees, ex-post, an $1\!-\!\frac{1}{q(\Bb,\Bs)}$ fraction of the efficient GFT, where $q(\Bb,\Bs)$ is the size of the most efficient trade (Theorem \ref{tr}).\footnote{If there is a trade with GFT of $0$, then there are efficient trades with different sizes. In this case trading according to the largest size will give full efficiency. }
As this mechanism gives no ex-ante guarantee (when $q(\Bb,\Bs)=1$), we create a \emph{hybrid mechanism} that runs the TR mechanism when $q(\Bb,\Bs)>1$ and run the RO mechanism with some constraints and conditional distributions otherwise. These constraints and conditionings of the distributions are needed both for incentive compatibility and for the ex-ante GFT guarantee. We now present this mechanism.

\begin{definition}[Hybrid Mechanism for Double Auctions]
Our \emph{hybrid mechanism} for double auctions is a direct revelation mechanism.
Given the reports $\Bb$ and $\Bs$ (that are assumed to be truthful), we
use $b_{(1)}$ to denote the buyer\footnote{Somewhat abusing notation, we use $b_{(1)}$ to refer both to this buyer and to his value, and similarly for other agents.} with maximum value (when breaking ties lexicographically by IDs), i.e., $b_{(1)}\geq b_i$ for every $i\in B$,
and use $b_{(2)}$ to denote the buyer with maximum value after removing buyer $b_{(1)}$.
Similarly, we use $s_{(1)}$ to denote the seller with minimal cost, and $s_{(2)}$ to denote the seller with the second-minimal cost.\footnote{Note that the maximal efficient set of trading agents is $\bigl\{s_{(1)},\ldots,s_{q(\Bb,\Bs)},b_{(1)},\ldots,b_{q(\Bb,\Bs)}\bigr\}$.}
The mechanism computes $q(\Bb,\Bs)$ and runs as follows.
\begin{itemize}
	\item If $q(\Bb,\Bs)\le1$,\footnote{Recall that in this case, if there is any trade with positive gains, then the maximal efficient set of trading agents is $\bigl\{s_{(1)},b_{(1)}\bigr\}$.} the mechanism computes the set of trading agents and payments by running the RO mechanism with SO parameters $\bar{s}=s_{(2)}$ and $D^B_{b_{(1)}}|_{\ge b_{(2)}}$ and BO parameters $\bar{b}=b_{(2)}$ and $D^S_{s_{(1)}}|_{\le s_{(2)}}$.\footnote{We note that in this case since $q(\Bb,\Bs)=1$, we have that $\bar{b}=b_{(2)}<s_{(2)}=\bar{s}$ and therefore indeed also $\bar{s}\ge\sup\supp\bigl(D^S_{s_{(1)}}|_{\le s_{(2)}}\bigr)$ and $\bar{b}\le\inf\supp\bigl(D^B_{b_{(1)}}|_{\ge b_{(2)}}\bigr)$.}
	\item If $q(\Bb,\Bs)\geq 2$, the mechanism computes the set of trading agents and payments by running the TR mechanism on $\Bb$ and $\Bs$.
\end{itemize}
\end{definition}

We will now sketch the intuition behind our choice, in the case where $q(\Bb,\Bs)=1$, of the constraints~$\bar{s},\bar{b}$ and the conditioned distributions $D^S_{s_{(1)}}|_{\le s_{(2)}}$ and $D^B_{b_{(1)}}|_{\ge b_{(2)}}$ for which the offered prices are optimized.
First, we would never want to allow $b_{(1)}$ to pay a price $p$ such that if $b_{(1)}$ had valuation~$p$ then she would not be in the first-best. This is since such a possibility would create an incentive for her to manipulate her bid if her valuation really were slightly higher than $p$ but still not high enough for her to be in the first-best: in this case, raising her bid would place her in the first-best, and she may end up paying~$p$, which would give her positive utility. So, we have to make sure that $b_{(1)}$ never offers, nor is ever offered, such a $p$ that is lower than $\bar{b}=b_{(2)}$. (In fact, the threshold bid of $b_{(1)}$ to be in the first-best is $\max\{b_{(2)},s_{(1)}\}$, but by definition of the RO mechanism, she would never pay less than $s_{(1)}$ as this would result in negative GFT, so we only need to make sure that she never pays less than $\bar{b}=b_{(2)}$.)
To make sure that $b_{(1)}$ never offers such a price $p$, we constrain her to offer at least $\bar{b}$ in the BO mechanism. To make sure that she is never offered such a price $p$ in the SO mechanism, we have $s_{(1)}$ optimize her offer under the assumption that the value of $b_{(1)}$ is drawn from $D^B_{b_{(1)}}|_{\ge\bar{b}}$, which is equivalent to disclosing to~$s_{(1)}$ that she has no point in offering a price lower than $\bar{b}$ since an offer of $\bar{b}$ will always be accepted. To see why the mechanism is truthful once we have set $\bar{b}$ (and $\bar{s}$) this way, consider the following hypothetical scenario. Say that after calculating that $q(\Bb,\Bs)=1$, the mechanism notifies $s_{(1)}$ and $b_{(1)}$ that they are the lowest-cost seller and highest-cost bidder, and furthermore notifies each of them of the values (and costs) of all other agents except the one that they are facing. In this case, the \emph{posterior distribution} of $s_{(1)}$ regarding $b_{(1)}$ is $D^B_{b_{(1)}}|_{\ge\max\{b_{(2)},s_{(1)}\}}$, so her best action is to optimize the price that she offers under this assumption, which is equivalent to optimizing the price that she offers for the distribution $D^B_{b_{(1)}}|_{\ge b_{(2)}}$ (but optimizing for the latter is easier to analyze, as it does not depend on the cost of $s_{(1)}$).

\begin{theorem}\label{thm:double-auction}
	For the double auction setting the above simple hybrid mechanism for double auctions is ex-post individually rational, Bayesian incentive compatible\footnote{Once again, since the allocation rule of the hybrid mechanism is ex-post monotone, by charging threshold prices we could strengthen the incentive-compatibility property from BIC to ex-post IC (while maintaining ex-post~IR), but then the weak-budget-balance guarantee would only hold ex-ante and not ex-post (similarly to the guarantee of \cref{bcfm}). Moreover, once we settle for ex-ante budget balance, we could get ex-ante \emph{strong} budget balance ``for free'' by equally dividing our ex-ante expected profits (assuming truthful bidding) among the agents, \citep[see, e.g.,][]{BrustleCFM17}.}, ex-post (direct trade) weakly budget balanced, and has both of the following efficiency guarantees:
	\begin{itemize}
		\item
It gets at least a $\nicefrac{1}{4}$-fraction of the efficient gains from trade ex-ante (second-best).
		\item
It gets at least a $\frac{q(\Bb,\Bs)-1}{q(\Bb,\Bs)}$-fraction of the efficient gains from trade ex-post (first-best).
Note that the mechanism is asymptotically efficient: as the trade size $q(\Bb,\Bs)$ goes to infinity, the fraction of the efficient gains from trade that it gets ex-post (first-best) goes to $1$.
	\end{itemize}
\end{theorem}

\subsection{Proof of Theorem~\ref{thm:double-auction}}

\begin{proof}[Proof of \cref{thm:double-auction}]
Recall that by \cref{tr,ro}, both the TR and the RO mechanisms are each ex-post IR, BIC, and ex-post (direct trade) weakly budget balanced.

\paragraph{Ex-post IR} Ex-post IR holds since both the TR and the RO mechanisms are ex-post IR.

\paragraph{Bayesian IC}
We will show that our hybrid mechanism is BIC for any buyer.\footnote{In fact, when our hybrid mechanism runs TR, then it is ex-post IC for every agent, and when a price~$p$ is offered by an agent in the RO mechanism, then our hybrid mechanism is Bayesian IC for the agent making the offer, and ex-post IC for all other agents including the agent who receives the offer.} A similar argument holds for truthfulness of the sellers.

We first claim that if a manipulation by a buyer does not change the choice of the mechanism that is run (TR or an instance of RO, where we consider each such instance to be a separate mechanism) by our hybrid mechanism, then it is nonbeneficial in expectation. For TR this follows since TR is ex-post IC.
To show this for RO,
we will show that that the region of the space of valuation/cost profiles where our hybrid mechanism for double auctions runs each instance of the RO mechanism can be partitioned into disjoint subsets where our hybrid mechanism is BIC on each such subset under the profile distribution conditioned upon being in that subset.

Fix a choice of the identity (but not the cost) of seller $s_{(1)}$ and the identity (but not the value) of bidder $b_{(1)}$, and fix a profile of costs and valuations for all other sellers and buyers (so in particular the cost $s_{(2)}$ and value $b_{(2)}$ are fixed). We first claim that either our hybrid mechanism runs the same instance of RO on all possible profiles $b,s$ that agree with these fixed choices, or does not run any instance of RO on any of these profiles. Indeed, if $s_{(2)}\le b_{(s)}$ (a conditioned fully determined by these fixed choices) then TR is run on all such profiles, and otherwise the RO mechanism with SO parameters $\bar{s}=s_{(2)}$ and $D^B_{b_{(1)}}|_{\ge b_{(2)}}$ and BO parameters $\bar{b}=b_{(2)}$ and $D^S_{s_{(1)}}|_{\le s_{(2)}}$ (note that all of these parameters are fully determined by the above fixed choices and do not depend on the cost $s_{(1)}$ or the value $b_{(1)}$) is run on all such profiles.

We will next show that our hybrid mechanism is BIC on the subset of all profiles that agree with these fixed choices. Note that when conditioning the distribution of all profiles to those that agree with such fixed choices, the cost of the seller $s_{(1)}$ (conditioned to agree with these fixed choices) is distributed precisely according to $s_{(1)}\sim D^S_{s_{(1)}}|_{\le s_{(2)}}$ and the value of the buyer $b_{(1)}$ (conditioned to agree with these fixed choices) is distributed precisely according to $b_{(1)}\sim D^B_{b_{(1)}}|_{\ge b_{(2)}}$. By \crefpart{ro}{properties}, we therefore have that our hybrid mechanism is BIC for the offering agent (and ex-post IC for any other agent) over all profiles that agree with these choices. We have therefore shown that if a manipulation by a buyer does not change the choice of the mechanism that is run by our hybrid mechanism, then it is nonbeneficial in expectation.

We now claim that a buyer who is in the efficient trading set cannot change the efficient trading set while remaining in this set. Indeed, to see that this is the case, suppose a buyer in the efficient trading misreports by adding $x$ (positive or negative) to his bid. The gains from trade from any trading set that includes this buyer therefore increase by $x$ (while the gains from trade of any other trading set remains the same); therefore, since we break ties in the same manner without and with the deviation, no other trading set that includes this buyer other than the true efficient trading set can ``become'' (as a result of the misreport) the new efficient trading set.

Since\ \ (1)~agents outside the efficient trading set never win,\ \ (2)~a buyer in the efficient trading set cannot change the efficient trading set while remaining in this set,\ \ (3)~the choice of the mechanism to run is completely determined by the efficient trading set and by the values/costs of the agents outside the efficient trading set, and\ \ (4)~a manipulation that does not change the choice of the mechanism to run is nonbeneficial in expectation,\ \ we conclude that there are no strategic opportunities (in expectation) for any buyer who is in the efficient trading set.

To complete the proof that our hybrid mechanism is BIC, it is therefore enough to show that there is no beneficial manipulation by a buyer who is not in the efficient trading set. We will in fact show that the mechanism is ex-post IC for such agents; we do so by considering several cases.
\begin{itemize}
\item
If $q(\Bb,\Bs)\geq 2$, then a buyer who is not in the efficient trading set cannot cause a move to $q(\Bb,\Bs)<2$. Any manipulation by such a buyer is therefore nonbeneficial since the TR mechanism (which is run prior to, and following, the manipulation) is ex-post IC.
\item
If $q(\Bb,\Bs)=1$, then we consider two possible manipulations by some buyer $b_{(j)}$ who is not in the efficient trading set (and is therefore not the true $b_{(1)}$):
\begin{itemize}
\item
First, consider a manipulation by $b_{(j)}$ that causes a move to $q(\Bb,\Bs)\ge2$ and causes her to win. We claim that in this case, this buyer, who was previously not in the efficient trading set, must pay at least her true value whenever she wins. Indeed, by definition of TR and since truly $q(\Bb,\Bs)=1$, since this buyer wins following the manipulation (and so is not reduced by the TR mechanism), she pays at least the original $b_{(1)}$, which is at least her true value. Therefore, she incurs non-positive utility.
\item
We next consider a manipulation by $b_{(j)}$ that maintains $q(\Bb,\Bs)=1$ and causes her to win (with some positive probability). We will show that whenever this buyer wins, she incurs non-positive utility. Since $q(\Bb,\Bs)=1$ is maintained following the manipulation, we must have that $b_{(j)}$ raised her bid to be higher than the original $b_{(1)}$, who is now in the role of $b_{(2)}$. By \crefpart{ro}{price-within-constraints}, if the manipulating buyer wins, then she pays at least the new $b_{(2)}$, i.e., the original $b_{(1)}$, which is at least her true value, and so she incurs non-positive utility.
\end{itemize}
\item
Finally, consider the case $q(\Bb,\Bs)=0$ and consider a manipulation by any buyer that causes her to win. Such a manipulation can only result in $q(\Bb,\Bs)=1$, so the manipulator, if she wins, trades with $s_{(1)}$, and by definition of RO and since this mechanism is ex-post IR for this seller, this buyer pays at least $s_{(1)}$. Since $q(\Bb,\Bs)=0$, we have that $s_{(1)}$ is larger than the true valuation of all buyers (including the manipulator), so the manipulator incurs negative utility whenever she wins.
\end{itemize}

\paragraph{Ex-post (direct trade) weak budget balance} Our hybrid mechanism is ex-post (direct trade) weakly budget balanced since the two mechanisms TR and  RO are both ex-post (direct trade) weakly budget balanced (the one is in fact ex-post (direct trade) strongly budget balanced).

\paragraph{Ex-post efficiency guarantee} When $q(\Bb,\Bs)=1$, then the guarantee vacuously holds, while when $q(\Bb,\Bs)\geq 2$, the guarantee follows from the same guarantee by the TR mechanism.

\paragraph{Ex-ante efficiency guarantee} We will show that for each valuation profile $(\Bb,\Bs)$, our hybrid mechanism achieves at least half of the gains from trade of the RVWM mechanism for the same valuation profile. Fix a valuation profile $(\Bb,\Bs)$. We consider several cases.
\begin{itemize}
\item
Consider the case where $s_{(2)}\le b_{(2)}$. Note that this is precisely the case where $q(\Bb,\Bs)\ge2$. In this case, our hybrid mechanism runs the TR mechanism, which by \cref{tr} guarantees at least a $\frac{q(\Bb,\Bs)-1}{q(\Bb,\Bs)}\ge\nicefrac{1}{2}$ fraction of the realized optimal gains from trade ex-post, and so at least a $\nicefrac{1}{2}$ fraction of the gains from trade of the RVWM mechanism.
\item
Consider the case where $s_{(1)}>b_{(1)}$. Note that this is precisely the case where $q(\Bb,\Bs)=0$. In this case, it is efficient to have no trade for $(\Bb,\Bs)$, and this is what both our hybrid mechanism and the RVWM mechanism do, so our hybrid mechanism has the same gains from trade as the RVWM mechanism.
\item
Consider the case where $b_{(2)} < s_{(2)}$, and in addition either $b_{(2)} \ge s_{(1)}$ or $s_{(2)} \le b_{(1)}$. In this case, since $q(\Bb,\Bs)=1$, we run the RO mechanism. By \crefpart{ro-trade-probability}{always-trade}, in this case $s_{(1)}$ and $b_{(1)}$ trade with probability at least $\nicefrac{1}{2}$, so our hybrid mechanism achieves at least a $\nicefrac{1}{2}$ fraction of the realized optimal gains from trade, and so at least a $\nicefrac{1}{2}$ fraction of the gains from trade of the RVWM mechanism.
\item
Finally, consider the case where $b_{(2)} < s_{(1)} \le b_{(1)} < s_{(2)}$. In this case, the only possible trading pair with positive gains is of $s_{(1)}$ with $b_{(1)}$, so if the RVWM mechanism achieves positive gains from trade, then it trades this pair with positive probability. By \cref{rvwm-offer}, in this case the GFT of the RVWM mechanism are therefore at least those of the RO mechanism with SO parameters $\infty$ (no constraint) and $D^B_{b(1)}$ (unconditioned distribution) and BO parameters $0$ (no constraint) and $D^S_{s_{(1)}}$ (unconditioned distribution) on that edge. Since $s_{(1)}>b_{(2)}=\bar{b}$ and $b_{(1)}<s_{(2)}=\bar{s}$, we have by \crefpart{ro-trade-probability}{restrictions-dont-matter} that the probability that trade occurs between $s_{(1)}$ and $b_{(1)}$ is at least as high in our hybrid mechanism (which runs the appropriate RO mechanism, constrained and conditioned) as it is in the unconstrained and unconditioned RO mechanism (that upper-bounds RVWM in this case). Therefore, in this case our hybrid mechanism achieves at least the gains from trade of the RVWM mechanism.
\end{itemize}
Combining all of the above, we have that
the expected gains from trade of our hybrid mechanism are at least a $\nicefrac{1}{2}$-fraction of those of the RVWM mechanism, and so by \cref{bcfm} at least a $\nicefrac{1}{4}$-fraction of the expected optimal gains from trade ex-ante (second-best).
\end{proof}

\subsubsection{Why The Proof of the Ex-Ante Guarantee Gives a Factor of \texorpdfstring{$\nicefrac{1}{4}$}{1/4} and Not~\texorpdfstring{$\nicefrac{1}{2}$}{1/2}}

Having read the proof of the ex-ante guarantee of \cref{thm:double-auction}, we note that at first glance, one may be tempted to consider the following na{\"i}ve adaptation of this proof into a ``proof'' of an ex-ante guarantee of $\nicefrac{1}{2}$ (rather than $\nicefrac{1}{4}$) of the second-best:
\begin{quote}
In each case analyzed above, the hybrid mechanism attains either at least the GFT of the RVWM mechanism, or at least half of the GFT of the first-best, which in turn is at least half of the GFT of the second-best. Since the GFT of the RVWM mechanism is in turn also at least half of the GFT of the second-best, we get that in either case our hybrid mechanism attains half of the GFT of the second-best.
\end{quote}
The problem with this ``proof'' is that it mixes ex-ante and ex-post guarantees. While indeed the hybrid mechanism attains, on each profile $(\Bb,\Bs)$, either at least the GFT of the RVWM mechanism or at least half of the GFT of the (first-best and therefore of the) second-best, it is wrong to assume that on each profile $(\Bb,\Bs)$ the GFT of the RVWM mechanism is at least half of the GFT of the second-best, as we only know that the expected GFT of the RVWM mechanism, over all profiles, is at least half of the expected GFT, over all profiles, of the second-best. In other words, it may hypothetically be that the RVWM mechanism performs poorly on the profiles on which our hybrid mechanism attains at least the GFT of the RVWM mechanism, and that the RVWM mechanism performs very well, surpassing half of the GFT of the second-best, and even half of the GFT of the first-best, on the profiles on which our hybrid mechanism attains at least half of the GFT of the first-best (so on average, the RVWM mechanism would indeed attain its guarantee), and in such a case, the above ``proof'' obviously fails.

\section{Main Results for Matching Markets}
\label{sec:matching}

In this section we will generalize the results of Section~\ref{sec:double-auctions} to matching markets.
Recall that a matching market is given by an undirected bipartite graph $G=(S,B,E)$ with nodes on one side representing the sellers and nodes on the other side representing the buyers, with edges indicating possible trades.
Recall that a profile $(\Bb,\Bs)$ assigns a value $b_i$ for each buyer $i\in B$ and a cost $s_j$ for each seller $j\in S$.

\subsection{A Trade Reduction Mechanism for Matching Markets}\label{sec:matching-tr}

We first present a generalized Trade Reduction mechanism for matching markets.
Like the Trade Reduction mechanism of \citet{McAfee92} for double-auctions, the \emph{Trade Reduction Mechanism for matching markets} that we define below picks a subset of the ``first-best'' trade, and determines the payments based on the values and costs of the agents that it removed from the first-best. The details are, however, more subtle than in the double-auction setting.

Recall from \cref{tr} that for every valuation profile $(\Bb,\Bs)$, the TR mechanism for double auctions attains GFT of at least a $1-\frac{1}{q(\Bb,\Bs)}$ fraction of $OPT(\Bb,\Bs)$. We note that giving the same guarantee for matching markets, with $q(\Bb,\Bs)$ remaining total the size of trade in the market, is not possible --- just consider a matching market that consists of two connected components, each a double auction. So, to phrase our TR mechanism for matching markets, we will first have to define some notation that will eventually help us phrase its GFT guarantee (which will still generalize the $1-\frac{1}{q(\Bb,\Bs)}$ of TR for double auctions, but in a slightly different way).

We say that the {\em classes of buyer $i$ and $i'$ are the same} if for any seller $j$ it holds that $(i,j)\in E$ if and only if  $(i',j)\in E$. Similarly, we define classes for sellers.\footnote{Note that the classes that we define depend neither on the values of the buyers nor on the costs of the sellers (nor on the distributions from which these values/costs are drawn).}
That is, two agents are of the same class if in any case that one of them can trade with some agent $x$, it also holds that the other agent can trade with agent $x$.
Thus, nodes in the graph can be partitioned into equivalent classes, where each equivalent class consists of all agents of some fixed class. Each such class either includes only buyers, or only sellers, but never both. Let $T_t$ denote the set of agents of class $t$.
For each class $t$ we denote by $q_t = q_t(M(\Bb,\Bs))$ the number of agents of class~$t$ that are matched in $M(\Bb,\Bs)$, that is
$q_t= |T_t\cap M(\Bb,\Bs)|$.
Additionally, we denote by $d_t= d_t(M(\Bb,\Bs))$ the number of distinct classes $t'$ such that there is an edge in $M(\Bb,\Bs)$ between an agent of class $t$ and an agent of class $t'$.

\begin{definition}[Trade Reduction Mechanism for Matching Markets]
Fix a matching market $G=(S,B,E)$.
The Trade Reduction mechanism for $G$ gets as input a profile $(\Bb,\Bs)$ and outputs an allocation and payments as follows.
\begin{itemize}
	\item Given profile $(\Bb,\Bs)$, let $M(\Bb,\Bs)$ be the ``first-best'' matching. Any agent not in $M(\Bb,\Bs)$ is marked as a loser and does not trade, paying 0.
	\item For each class $t$,  recall that $q_t$ is the number of agent of class $t$ that are matched in $M(\Bb,\Bs)$ and $d_t$ is the number of different classes that trade with agents of class $t$ in $M(\Bb,\Bs)$.
	\begin{itemize}
	\item
	For each buyer class $t$, the set of trading buyers will be the set of $q_t-d_t$ highest-value buyers of class $t$ (breaking ties lexicographically by IDs).\footnote{Note that the number of trading buyers is non-negative, as for every class $t$ it holds that $q_t\geq d_t$.}
	We say that $d_t$ buyers of class $t$ were \emph{reduced}.
	Each buyer of class $t$ pays the highest value reported by any reduced buyer of class $t$.   	
	\item
	For each seller class $t$, the set of trading sellers will be the set of $q_t-d_t$ lowest-cost sellers of class $t$ (breaking ties lexicographically by IDs).\footnote{Note that the number of trading sellers is non-negative, as for every class $t$ it holds that $q_t\geq d_t$.}
	We say that $d_t$ sellers of class $t$ were \emph{reduced}.
	Each buyer of class $t$ is paid the lowest cost reported by any reduced seller of class $t$.   	
	\end{itemize}
\end{itemize}
We denote the set of agents that are trading under this mechanism by $\mathit{TR}(\Bb,\Bs)$.
\end{definition}

The following \lcnamecref{thm:matching-TR} presents the properties of the Trade Reduction Mechanism for matching markets. In particular, it shows that the mechanism provides some ex-post GFT guarantees which is a function of the maximum weight matching $M(\Bb,\Bs)$. As with the TR mechanism for double auctions, this mechanism does not provide any ex-ante guarantees, though, even with respect to the second-best mechanism. In particular, with a single trade in $OPT(\Bb,\Bs)$, there will be no trade in this mechanism.
\begin{theorem}\label{thm:matching-TR}
	The Trade Reduction Mechanism for matching markets is ex-post IR, ex-post IC, ex-post (direct trade) weakly budget balanced,
	and for any $(\Bb,\Bs)$ the fraction of the gains from trade of $OPT(\Bb,\Bs)$ that it attains is at least
	$\min\bigl\{1- \frac{d_t}{q_t} ~\big|~ \text{class $t$  s.t.\ $q_t>0$}\bigr\}$.\footnote{Note that $d_t, q_t$ and $r_{t,t'}$ are all function of $M(\Bb,\Bs)$, so they are functions of the profile $(\Bb,\Bs)$. This is similar to $q(\Bb,\Bs)$ being a function of $(\Bb,\Bs)$ for Trade Reduction in double auctions.}
\end{theorem}

We note that guarantee from \cref{thm:matching-TR}  of the TR mechanism attaining a fraction-of-$OPT(\Bb,\Bs)$ of at least
$\alpha(\Bb,\Bs)=\min\bigl\{1- \tfrac{d_t}{q_t} ~\big|~ \text{class $t$  s.t.\ $q_t>0$}\bigr\}$
coincides in the double-auction setting with the guarantee of at least $1-\frac{1}{q(\Bb,\Bs)}$ from \cref{tr}, and naturally generalizes it. Another generalization for matching markets of the fraction $1-\frac{1}{q(\Bb,\Bs)}$ that one may find natural, which also coincides with it in the double-auction setting, is
$\beta(\Bb,\Bs)=\min\bigl\{ 1- \tfrac{1}{r_{t,t'}} ~\big|~\text{classes $(t,t')$ s.t.\ $ r_{t,t'}>0$}\bigr\},$
where for any buyers' class $t$ and sellers' class $t'$, we use $r_{t,t'}$ to denote the number of buyers of class $t$ that are matched with sellers of class $t'$ in $M(\Bb,\Bs)$. While this alternative generalization is conceptually interesting in its own right, we in fact show that for every valuation profile it holds that $\beta(\Bb,\Bs)\le\alpha(\Bb,\Bs)$, and so a GFT guarantee of $\beta(\Bb,\Bs)$ follows from the GFT guarantee of $\alpha(\Bb,\Bs)$ from \cref{thm:matching-TR}:

\begin{corollary}\label[corollary]{cor:matching-TR-beta}
For any $(\Bb,\Bs)$, the fraction of the GFT of $OPT(\Bb,\Bs)$ that the TR mechanism for matching markets attains is at least
$\min\bigl\{ 1- \frac{1}{r_{t,t'}} ~\big|~\text{classes $(t,t')$ s.t.\ $ r_{t,t'}>0$}\bigr\}.$
\end{corollary}

\noindent
The proofs of \cref{thm:matching-TR,cor:matching-TR-beta} are given in \cref{app:matching-tr}.

\subsection{The Offering Mechanism for Matching Markets}\label{sec:matching-offering}

Before defining our hybrid mechanism for matching markets, we first define an offering mechanism for matching markets, analogous to the specific instance of the RO mechanism (including the specific offer constraints and conditioned distributions) that our hybrid mechanism for double auctions runs whenever $q(\Bb,\Bs)=1$ in that setting. In this mechanism, agents not in  $M(\Bb,\Bs)$ never trade, and agent in a pair $(i,j)\in M(\Bb,\Bs)$ either trades in that pair or does not trade at all. This mechanism is defined as follows.

\begin{definition}[Offering Mechanism for Matching Markets]
The mechanism iterates over all edges $(i,j)\in M(\Bb,\Bs)$, and for each such edge acts as follows.
\begin{itemize}
\item
Let $\bar{s}=\bar{s}_{(i,j)}(\Bb,\Bs)$ be the minimal bid of buyer $i$ such that any higher bid causes $i$ to be in the first-best in the market $(S\setminus\{j\},B)$, i.e., the market without seller $j$. We set $\bar{s}=\infty$ if no such bid exists.
\item
Let $\bar{b}=\bar{b}_{(i,j)}(\Bb,\Bs)$ be the maximal bid (reported cost) of seller $j$ that causes $j$ to be in the first-best in the market $(S,B\setminus\{i\})$, i.e., the market without buyer $i$. We set $\bar{b}=0$ if no such bid exists.
\end{itemize}
Now, to decide whether trade occurs between $i$ and $j$ and at which price, run the RO mechanism on this edge with SO parameters $\bar{s}$ and $D^B_i|_{\ge\bar{b}}$ and BO parameters $\bar{b}$ and $D^S_j|_{\le\bar{s}}$.
\end{definition}

We note that the above offer constraints $\bar{s}$ and $\bar{b}$ precisely generalize the offer constraints from our hybrid mechanism for double auctions from \cref{sec:double-auctions}. Indeed, in a double auction, the minimal bid of buyer $b_{(1)}$ that causes her to be in the first-best in the market $(S\setminus\{s_{(1)}\},B)$ without $s_{(1)}$ is $\max\{b_{(2)},s_{(2)}\}$, and when $q(s,b)=1$ in the double-auctions setting (this is the case where we run the RO mechanism) it must be that $b_{(2)}<s_{(2)}$ and so $\max\{b_{(2)},s_{(2)}\}=s_{(2)}$, which how we set the constraint $\bar{s}$ in that mechanism. The choice of $\bar{b}$ is similar. The careful definition of $\bar{s}$ and $\bar{b}$ above guarantees the two properties of these thresholds that our double-auction constraints readily satisfied: first, both $\bar{b}$ and~$\bar{s}$ are completely independent of $b_i$ and of $s_j$, and second, as we will see in our analysis, $\bar{b}$ coincides with the minimal winning bid of \emph{buyer $i$} in the \emph{original} market whenever this constraint is binding. (And similarly for $\bar{s}$ and seller $j$.)

To show that the Offering Mechanism is well-defined, we have to make sure that the SO and BO parameters that we specify for the RO mechanism meet the conditions imposed in the definition of that mechanism. The following \lcnamecref{offering-good-parameters} does precisely this.

\begin{lemma}\label[lemma]{offering-good-parameters}
For every $(i,j)\in M(\Bb,\Bs)$, it holds that\ \ \emph{(1)}~$\bar{s}\ge s_j$,\ \ \emph{(2)}~$\bar{b}\le b_i$, and\ \ \emph{(3)}~$\bar{s}\ge\bar{b}$.
\end{lemma}

\noindent
We next prove that the Offering Mechanism is truthful, budget balanced, and has an ex-ante guarantee.

\begin{theorem}\label{offering}
The Offering Mechanism is BIC, ex-post IR, ex-post (direct trade) strongly budget balanced, and ex-ante guarantees at least a $\nicefrac{1}{4}$ of the expected GFT
	of the second-best mechanism.
\end{theorem}

The proofs of \cref{offering-good-parameters,offering} is given in \cref{app:matching-offering}. As noted in the introduction, proving the ex-ante guarantee of the Offering Mechanism for matching markets is the most technically challenging part of our analysis. The main ideas behind this proof are surveyed in \cref{sec:ex-ante}.

\subsection{The Hybrid Mechanism for Matching Markets}\label{sec:matching-hybrid}

We are now ready to define our hybrid mechanism for matching markets. It combines the TR mechanism and the Offering Mechanism in a proper way. We note that for double auctions, the mechanism defined below reduces precisely to our hybrid mechanism for double auctions from \cref{sec:double-auctions}.

\begin{definition}[Hybrid Mechanism for Matching Markets]
Let $G=(S,B,E)$ be the constraints graph.
Our \emph{hybrid mechanism} is a direct revelation mechanism.
Given the the reports $(\Bb,\Bs)$ (which is assumed to be truthful), the mechanism computes $M(\Bb,\Bs)$ and $\alpha(\Bb,\Bs) = \min\Bigl\{1- \frac{d_t}{q_t} ~\Big|~ \text{class $t$ s.t.\ $q_t>0$} \Bigr\}$ and runs as follows.
\begin{itemize}
\item
If $\alpha(\Bb,\Bs)\ge\nicefrac{1}{2}$, the mechanism computes the set of trading agents and payments by running the Trade Reduction Mechanism for matching markets defined above.
\item
 Otherwise, the mechanism computes the set of trading agents and payments by running the Offering Mechanism for matching markets defined above.
\end{itemize}	
\end{definition}

\noindent
We are now ready to formally state the main result of this paper.

\begin{theorem}\label{thm:matching-main}
	The Hybrid Mechanism for matching markets is ex-post IR, BIC\footnote{As in the double-auction setting, since the allocation rule of the hybrid mechanism for matching markets is ex-post monotone, by charging threshold prices we could strengthen the incentive-compatibility property from BIC to ex-post~IC (while maintaining ex-post~IR), but then the weak-budget-balance guarantee would only hold ex-ante and not ex-post (similarly to the guarantee of \cref{bcfm}). Moreover, once we settle for ex-ante budget balance, we could get ex-ante \emph{strong} budget balance ``for free'' by equally dividing our ex-ante expected profits (assuming truthful bidding) among the agents, \citep[see, e.g.,][]{BrustleCFM17}.}. and ex-post (direct trade) weakly budget balanced, which satisfies both of the following.
	\begin{itemize}
		\item
		The expected GFT of this mechanism are at least $\nicefrac{1}{4}$ of those of the second-best mechanism.
		\item
		For any $(\Bb,\Bs)$ with $\alpha(\Bb,\Bs)\geq\nicefrac{1}{2}$, the fraction of the gains from trade of $OPT(\Bb,\Bs)$ that this mechanism attains is at least $\alpha(\Bb,\Bs)\ge\nicefrac{1}{2}$.
	\end{itemize}
\end{theorem}

The hybrid mechanism for matching markets inherits from the Trade Reduction mechanism for matching markets also the ex-post guarantee of \cref{cor:matching-TR-beta}:

\begin{corollary}\label[corollary]{cor:matching-main-beta}
Let $\beta(\Bb,\Bs)=\min\bigl\{ 1- \frac{1}{r_{t,t'}} ~\big|~\text{classes $(t,t')$ s.t.\ $ r_{t,t'}>0$}\bigr\}$.
For any $(\Bb,\Bs)$ with $\beta(\Bb,\Bs)\geq\nicefrac{1}{2}$, the fraction of the gains from trade of $OPT(\Bb,\Bs)$ that the Hybrid Mechanism for matching markets attains is at least $\beta(\Bb,\Bs)\ge\nicefrac{1}{2}$.
\end{corollary}

\noindent
The proofs of \cref{thm:matching-main,cor:matching-main-beta} are given in \cref{app:matching-hybrid}.

\section{Sketch of the Proof of Ex-Ante Guarantee of the\texorpdfstring{\\}{ }Offering Mechanism for Matching Markets}
\label{sec:ex-ante}

In this section, we sketch the proof of the ex-ante guarantee of the Offering Mechanism, which has been stated in \cref{offering}. The full proof is relegated to \cref{app:matching-offering}.
To prove
that the Offering Mechanism ex-ante guarantees at least a $\nicefrac{1}{4}$-fraction of the gains from trade of the second-best mechanism,
we compare the Offering Mechanism to the RVWM mechanism of \citet{BrustleCFM17}. Due to \cref{bcfm}, it suffices to show the following lemma.

\begin{lemma}\label[lemma]{offering-bcfm-pointwise}
	For any valuation profile $(\Bb,\Bs)$, the gains from trade of the Offering Mechanism for matching markets is at least half of the gains from trade of the RVWM mechanism for that profile.
\end{lemma}

We first provide the intuition behind \cref{offering-bcfm-pointwise}. Fix a valuation profile $(\Bb,\Bs)$. Let $M_1^*=M_1^*(\Bb,\Bs)$ be the maximum-weight matching of $G$ when edge weight is $\tilde{\varphi}_i(b_i)-s_j$ for $(i,j)\in E$ and $M_2^*=M_2^*(\Bb,\Bs)$ be the maximum-weight matching of $G$ when edge weight is $b_i-\tilde{\tau}_j(s_j)$ for $(i,j)\in E$.\footnote{Recall that $\tilde{\varphi}_i$ and $\tilde{\tau}_j$ are the ironed virtual value functions of buyer $i$ and seller $j$, respectively.%
} Recall that the RVWM mechanism runs the Generalized Seller Offering Mechanism (GSOM)
with probability  $\nicefrac{1}{2}$ and in that case obtains the GFT of the matching $M_1^*$, and it runs the Generalized Buyer Offering Mechanism (GBOM) with probability $\nicefrac{1}{2}$ and in that case obtains the GFT of the matching $M_2^*$.
It suffices to show that the GFT of each of the two matchings $M_1^*$ and $M_2^*$ can be bounded by twice the GFT of the Offering Mechanism for the valuation profile $(\Bb,\Bs)$. We will show how to bound the GFT of $M_1^*$. A similar argument can bound the GFT of $M_2^*$.

Consider the first-best matching $M=M(\Bb,\Bs)$ together with the matching $M_1^*$. Each connected component of the union of the two matchings $M\cup M_1^*$ is either a maximal alternating path\footnote{A path is called an \emph{alternating path} if the edges of the path alternate between the two matchings. A path is \emph{maximal} if it is not a subpath of any other path.} or an alternating cycle\footnote{An \emph{alternating cycle} is an alternating path whose two endpoints coincide.}. We will show that all alternating cycles consist of two edges between the same seller and buyer (see \cref{fig:three-cases} (a) ) due to our tie-breaking rules (proved in \cref{cor:cycles}) and that the GFT of the Offering Mechanism from that buyer-seller pair is at least the GFT of the RVWM mechanism from that pair. For an alternating path, we will consider the cases of an even or an odd number of edges of the alternating path separately, and show that in either case, the GFT of our Offering Mechanism from the path is at least half of the GFT of the matching $M_1^*$ from that path. Given the fact that $M$ and $M_1^*$ are each a maximum-weight matching w.r.t.\ the edge weights $b_i-s_j$ and $\tilde{\varphi}_i(b_i)-s_j$ respectively, we prove that any maximal alternating path that is not a cycle, starts with a buyer and an edge from $M$. See \cref{cor:buyer-first} for more details. \cref{fig:three-cases} illustrates the three different cases in our proof.

\begin{figure}[t]
\begin{center}
\begin{tabular}{ccc}
\includegraphics[scale=0.55]{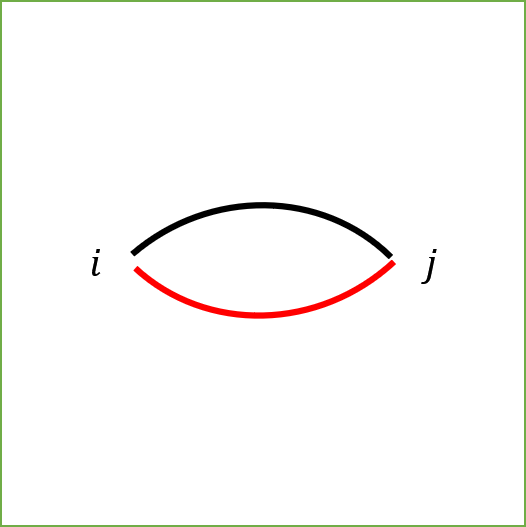} & \includegraphics[scale=0.55]{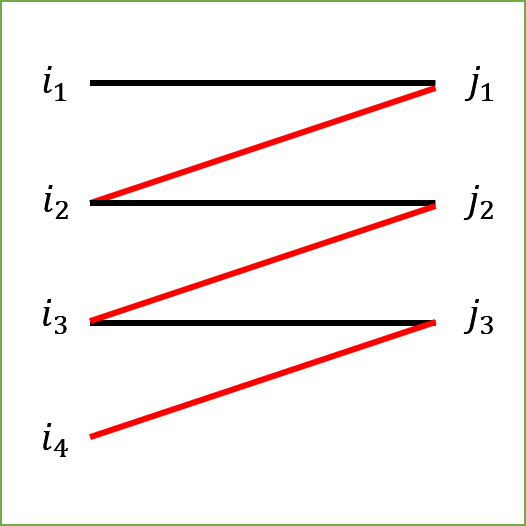} &  \includegraphics[scale=0.55]{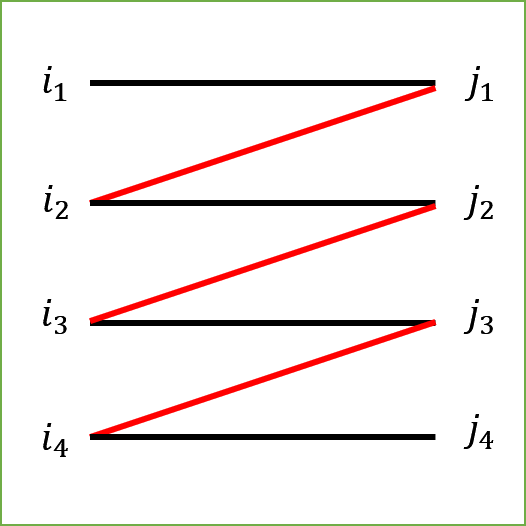}\\
(a) & (b) & (c)\\
\end{tabular}
\end{center}
\caption{Example of the three different cases considered in the proof of the ex-ante guarantee: (a) alternating cycle; (b) maximal alternating path with even number of edges; (c) maximal alternating path with odd number of edges.}
\label{fig:three-cases}
\end{figure}

The following lemma plays a central role in our proof of \cref{offering-bcfm-pointwise}. It provides a sufficient condition for a buyer-seller pair to trade in that mechanism.
\begin{lemma}\label[lemma]{lem:sufficient-trade-offering}
	Fix valuation profile $(\Bb,\Bs)$. For every $(i,j)\in M(\Bb,\Bs)$, if $j$ is in $M_{-i}(\Bb,\Bs)$ then buyer $i$ will trade with seller $j$ in the BO Mechanism, and if $i$ is in $M_{-j}(\Bb,\Bs)$ then buyer $i$ will trade with seller $j$ in the SO Mechanism. Thus, in either case $i$ and $j$ will trade with probability at least $\nicefrac{1}{2}$ in the Offering Mechanism.
\end{lemma}

The next lemma shows that any seller who is not at the end of any alternating path, must still be in the first-best matching if we remove the buyer that is matched to her.

\begin{lemma}\label[lemma]{lem:even-odd-path}
	Let $A$ be an acyclic maximal alternating path of $M(\Bb,\Bs)\cup M_1^*(\Bb,\Bs)$.
	For every seller $j\in A$ who is not at the end of the path,
	it holds that $j\in M_{-i}(\Bb,\Bs)$, where $i$ is the buyer such that $(i,j)\in M(\Bb,\Bs)$.
\end{lemma}

Combining \cref{lem:even-odd-path,lem:sufficient-trade-offering}, we show that all sellers in a maximal alternating path of even length will trade in the BO mechanism with the buyers that are matched to them in the first-best matching.

Next, we consider maximal alternating paths with odd length and present another useful characterization. We assume w.l.o.g.\ that any  maximal alternating path of  $M\cup M_1^*$ starts with a buyer and an edge in $M$ (by Corollary \ref{cor:buyer-first}).
Let
$GFT_{M'}(U)$ be the GFT of all edges of $M'$ that are contained in $U$.

\begin{lemma}\label[lemma]{lem:odd-path-extra-edge}
	For $K>3$, let $A=(i_1j_1i_2j_2...i_{L-1}j_{L-1}i_Lj_L)$, be a maximal alternating path of odd number of edges of $M\cup M_1^*$ with
	$i_l$ denoting buyers and $j_l$ denoting sellers, and with $(i_1,j_1)\in M$.
It holds that
	\begin{itemize}
		\item if $b_{i_L}>b_{i_1}$ then  $i_L\in M_{-j_L}$.
		\item if $b_{i_L}\leq b_{i_1}$ then $GFT_M(A\setminus \{i_L,j_L\} )\geq GFT_{M_1^*}(A) $.
	\end{itemize}
\end{lemma}

\begin{proof}[Proof of \cref{offering-bcfm-pointwise}]
By \cref{cor:cycles}, any alternating cycle of $M$ and $M_1^*$ has only two (identical) undirected edges $(i,j)\in M\cap M_1^*$. If $j\in M_{-i}(\Bb,\Bs)$ or $i\in M_{-j}(\Bb,\Bs)$, by \cref{lem:sufficient-trade-offering} buyer $i$ will trade with seller $j$ in the Offering Mechanism with probability at least $\nicefrac{1}{2}$, which obtains at least half of the GFT that the RVWM mechanism obtains on $(i,j)\in M_1^*$ when the profile is $(\Bb,\Bs)$.
Otherwise, since $i\notin M_{-j}$ we have that $\bar{s} \ge b_i$, and since $j\notin M_{-i}$ we have that $\bar{b}\le s_j$.
	Since trade occurs with positive probability on $(i,j)$ in the RVWM mechanism, then
similarly to the double-auction case, by \cref{rvwm-offer} and \crefpart{ro-trade-probability}{restrictions-dont-matter},
our Offering Mechanism achieves at least the gains from trade of the RVWM mechanism on this edge (and therefore, on any alternating cycle).
	
Consider any maximal alternating path of even number of edges.
By \cref{lem:sufficient-trade-offering,lem:even-odd-path}, every pair $(i,j)\in M$ trades in the BO mechanism, so whenever the BO
 mechanism runs, the maximal GFT (first-best) of the agents in the alternating path, which is at least the GFT of $M_1^*$ from these agents, is obtained. The Offering Mechanism runs the BO mechanism is probability $\nicefrac{1}{2}$, so in expectation it obtains at least $\nicefrac{1}{2}$ the GFT of $M_1^*$ from this path.

Now consider any maximal alternating path $(i_1j_1i_2j_2...i_{L-1}j_{L-1}i_Lj_L)$ of odd number of (at least 3)\footnote{If there is a single edge, then it is only in $M$. We only need to cover edges in $M_1^*$.} edges, which starts with buyer $i_1$ and an edge from $M$. Here $L\geq 2$. By \cref{lem:sufficient-trade-offering,lem:even-odd-path}, for every $l=1,2,\ldots,L-1$, buyer $i_l$ will trade with seller $j_l$ in the BO mechanism. If $b_{i_1}\geq b_{i_L}$, the claim holds since by \cref{lem:odd-path-extra-edge} the GFT of $M_1^*$ from this path is at most the GFT of the first $L\!-\!1$ pairs in the first-best matching $M$, and all these $L\!-\!1$ pairs will be traded in the BO mechanism, which happens with probability~$\nicefrac{1}{2}$.

If $b_{i_1}<b_{i_L}$, by \cref{lem:sufficient-trade-offering,lem:odd-path-extra-edge}, buyer $i_L$ will trade with seller $j_L$ in the SO mechanism, which happens with probability $\nicefrac{1}{2}$. Therefore, every pair $(i,j)\in M$ is traded with probability at least $\nicefrac{1}{2}$. This obtains half the maximal GFT (first-best) of this path, which is at least half the GFT of $M_1^*$ in this path.

Similarly, we can show that the Offering Mechanisms obtains at least $\nicefrac{1}{2}$ of the GFT of $M_2^*$. Since the expected GFT of the RVWM mechanism is the average GFT of $M_1^*$ and $M_2^*$, we conclude that the Offering Mechanisms obtains at least $\nicefrac{1}{2}$ the GFT of the RVWM mechanism.
\end{proof}

\section{Conclusion}
\label{sec:conclusion}

One of the biggest pushbacks against constant-approximation mechanisms is that while they provide some worst-case guarantee, they often do not provide any guarantee for significantly better performance ``when the instances are easy to handle''. We believe that a mechanism that not only provides a worst-case guarantee, but also provides a guarantee of performing very well on ``easy instances'' is much more appealing and more likely to be used. In our setting, we implement this agenda by postulating that ``nice instances'' are large-market instances (for some formal sense of ``large''), and we are able to achieve the best of both worlds: a guaranteed constant approximation on one hand, and asymptotic optimality when the markets are large on the other hand. We believe that presenting similar results in other settings is an interesting research direction.

Mechanisms with such ``worse-case \emph{and} best-case guarantees'' are of particular appeal when the social planner needs to fix the mechanism well before the exact market characteristics are known, for example, when the mechanism is defined by some laws or regulations (e.g., FCC auctions) that are fixed well in advance.

\section*{Acknowledgements}
Yang Cai is supported by NSERC Discovery RGPIN-2015-06127 and FRQNT 2017-NC-198956.
Yannai Gonczarowski is supported by the Adams Fellowship Program of the Israel Academy of Sciences and Humanities;
his work is supported by ISF grant 1435/14 administered by the Israeli Academy of Sciences and by Israel-USA Bi-national Science Foundation (BSF) grant number 2014389;
this project has received funding from the European Research Council (ERC) under the European Union's
Horizon 2020 research and innovation programme (grant agreement No 740282).
Mingfei Zhao is supported by NSERC Discovery RGPIN-2015-06127 and Richard H. Tomlinson Doctoral Fellowship.

\bibliographystyle{plainnat}
\bibliography{bib}

\begin{thebibliography}{12}
\providecommand{\natexlab}[1]{#1}
\providecommand{\url}[1]{\texttt{#1}}
\expandafter\ifx\csname urlstyle\endcsname\relax
  \providecommand{\doi}[1]{doi: #1}\else
  \providecommand{\doi}{doi: \begingroup \urlstyle{rm}\Url}\fi

\bibitem[Babaioff and Walsh(2005)]{BabaioffW05}
Moshe Babaioff and William~E. Walsh.
\newblock Incentive-compatible, budget-balanced, yet highly efficient auctions
  for supply chain formation.
\newblock \emph{Decision Support Systems}, 39\penalty0 (1):\penalty0 123--149,
  2005.

\bibitem[Babaioff et~al.(2009)Babaioff, Nisan, and Pavlov]{BabaioffNP09}
Moshe Babaioff, Noam Nisan, and Elan Pavlov.
\newblock Mechanisms for a spatially distributed market.
\newblock \emph{Games and Economic Behavior}, 66\penalty0 (2):\penalty0
  660--684, 2009.

\bibitem[Blumrosen and Dobzinski(2016)]{BlumrosenD16}
Liad Blumrosen and Shahar Dobzinski.
\newblock (almost) efficient mechanisms for bilateral trading.
\newblock arXiv preprint arXiv:1604.04876, 2016.

\bibitem[Blumrosen and Mizrahi(2016)]{BlumrosenM16}
Liad Blumrosen and Yehonatan Mizrahi.
\newblock Approximating gains-from-trade in bilateral trading.
\newblock In \emph{Proceedings of the 12th Conference on Web and Internet
  Economics (WINE)}, pages 400--413, 2016.

\bibitem[Brustle et~al.(2017)Brustle, Cai, Wu, and Zhao]{BrustleCFM17}
Johannes Brustle, Yang Cai, Fa~Wu, and Mingfei Zhao.
\newblock Approximating gains from trade in two-sided markets via simple
  mechanisms.
\newblock In \emph{Proceedings of the 18th ACM Conference on Economics and
  Computation (EC)}, pages 589--590, 2017.
\newblock ISBN 978-1-4503-4527-9.

\bibitem[Colini{-}Baldeschi et~al.(2016)Colini{-}Baldeschi, de~Keijzer,
  Leonardi, and Turchetta]{Colini-Baldeschi16}
Riccardo Colini{-}Baldeschi, Bart de~Keijzer, Stefano Leonardi, and Stefano
  Turchetta.
\newblock Approximately efficient double auctions with strong budget balance.
\newblock In \emph{Proceedings of the 27th Annual ACM-SIAM Symposium on
  Discrete Algorithms (SODA)}, pages 1424--1443, 2016.

\bibitem[Colini-Baldeschi et~al.(2017)Colini-Baldeschi, Goldberg, de~Keijzer,
  Leonardi, Roughgarden, and Turchetta]{Colini-Baldeschi17}
Riccardo Colini-Baldeschi, Paul~W. Goldberg, Bart de~Keijzer, Stefano Leonardi,
  Tim Roughgarden, and Stefano Turchetta.
\newblock Approximately efficient two-sided combinatorial auctions.
\newblock In \emph{Proceedings of the 18th ACM Conference on Economics and
  Computation (EC)}, pages 591--608, 2017.
\newblock ISBN 978-1-4503-4527-9.

\bibitem[D{\"u}tting et~al.(2017)D{\"u}tting, Talgam-Cohen, and
  Roughgarden]{DuttingTR17}
Paul D{\"u}tting, Inbal Talgam-Cohen, and Tim Roughgarden.
\newblock Modularity and greed in double auctions.
\newblock \emph{Games and Economic Behavior}, 105\penalty0 (C):\penalty0
  59--83, 2017.

\bibitem[Gul and Stacchetti(1999)]{GulS99}
Faruk Gul and Ennio Stacchetti.
\newblock {Walrasian Equilibrium with Gross Substitutes}.
\newblock \emph{Journal of Economic Theory}, 87\penalty0 (1):\penalty0 95--124,
  1999.

\bibitem[McAfee(1992)]{McAfee92}
R.~Preston McAfee.
\newblock A dominant strategy double auction.
\newblock \emph{Journal of Economic Theory}, 56\penalty0 (2):\penalty0
  434--450, 1992.

\bibitem[Myerson(1981)]{Myerson81}
Roger~B. Myerson.
\newblock {Optimal Auction Design}.
\newblock \emph{Mathematics of Operations Research}, 6\penalty0 (1):\penalty0
  58--73, 1981.

\bibitem[Myerson and Satterthwaite(1983)]{MyersonS83}
Roger~B. Myerson and Mark~A. Satterthwaite.
\newblock Efficient mechanisms for bilateral trading.
\newblock \emph{Journal of Economic Theory}, 29\penalty0 (2):\penalty0
  265--281, 1983.

\end{thebibliography}

\appendix

\section{Proof of Examples~\ref{example:rvwm-not-eff} through~\ref{example:rvwm-switch-not-monotone}}
\label{app:examples}

\begin{proof}[Proof of \cref{example:rvwm-not-eff}]
First consider the first-best mechanism. Fix buyer $i$. We'll prove that if her value $b_i>\nicefrac{1}{2}+\epsilon$ for some small $\epsilon$ that will be determined later, she will trade in the first-best mechanism with high probability. Let $\mathcal{S}$ be the number of sellers whose cost is smaller than $\nicefrac{1}{2}+\epsilon$, and let $\mathcal{B}$ be the number of buyers whose value is larger than $\nicefrac{1}{2}+\epsilon$. Notice that there are at most $\mathcal{B}$ buyers with value larger than $i$. Thus if $\mathcal{B}<\mathcal{S}$, all buyers with value $>\nicefrac{1}{2}+\epsilon$, including $i$, must have traded with a seller whose cost is smaller than $\nicefrac{1}{2}+\epsilon$ in the first-best mechanism.

Both $\mathcal{S}$ and $\mathcal{B}$ are the sum of independent Bernoulli random variables, with expectation $(\nicefrac{1}{2}+\epsilon)n$ and $(\nicefrac{1}{2}-\epsilon)(n-1)$ respectively. By Chernoff bound,
$$\Pr\left[\mathcal{S}<(1-\epsilon)\left(\frac{1}{2}+\epsilon\right)n\right]\leq exp\left(-\frac{\epsilon^2}{2}\cdot\left(\frac{1}{2}+\epsilon\right)n\right)\leq exp\left(-\frac{1}{12}\epsilon^2n\right),$$
$$\Pr\left[\mathcal{B}>(1+\epsilon)\left(\frac{1}{2}-\epsilon\right)(n-1)\right]\leq exp\left(-\frac{\epsilon^2}{3}\cdot\left(\frac{1}{2}-\epsilon\right)(n-1)\right)\leq exp\left(-\frac{1}{12}\epsilon^2n\right).$$

With probability at least $\left(1-exp\left(-\frac{1}{12}\epsilon^2n\right)\right)^2$,

$$\mathcal{S}\geq (1-\epsilon)\left(\frac{1}{2}+\epsilon\right)n>(1+\epsilon)\left(\frac{1}{2}-\epsilon\right)(n-1)\geq \mathcal{B}.$$

In other words, if $b_i>\nicefrac{1}{2}+\epsilon$, buyer $i$ will trade in the first-best mechanism with probability at least $\left(1-exp\left(-\frac{1}{12}\epsilon^2n\right)\right)^2$, taking expectation over all other agents' types.

Choose $\epsilon=n^{-\nicefrac{1}{3}}$. The expected GFT contributed by buyer $i$ is at least

$$\int_{\frac{1}{2}+n^{-\frac{1}{3}}}^1b_idb_i\cdot\left(1-exp\left(-\frac{1}{12}n^{\frac{1}{3}}\right)\right)^2=\frac{3}{8}+o(1).$$

By linearity of expectation and the fact that all buyers are i.i.d., the expected GFT contributed by all buyers is at most $\nicefrac{3n}{8}+o(n)$.

Similarly for every seller $j$, if her cost $s_j>\nicefrac{1}{2}+\epsilon$ for some small $\epsilon$, she can only trade in the first-best mechanism with small probability. The expected GFT contributed by all sellers (a negative term) is at least $-\nicefrac{n}{8}+o(n)$. Thus the first-best mechanism obtains GFT at least $\nicefrac{n}{4}+o(n)$. The second-best mechanism gets GFT that is in expectation at least the expected GFT of the Trade Reduction mechanism, and as the GFT of the TR mechanism is at least the GFT of the first best minus $1$,  the second-best mechanism obtains GFT at least $\nicefrac{n}{4}+o(n)$ as well.

For the mechanism of \citet{BrustleCFM17}, we first refer the reader to Appendix~\ref{sec:matching-preliminaries} for the formal definition of the mechanism. In the mechanism, with probability a half, the mechanism implements GSOM that finds all the efficient trade based on buyers' virtual value and sellers' cost. And probability a half, it implements GBOM that finds all the efficient trade based on buyers' value and seller' virtual cost. We will only give the proof that the expected GFT of GSOM is at most $\nicefrac{2n}{9}+o(n)$. An analogous proof shows that the expected GFT of GBOM is at most $\nicefrac{2n}{9}+o(n)$.

For each buyer whose value is drawn from uniform distribution $[0,1]$, her virtual value follows the uniform distribution $[-1,1]$. Fix buyer $i$, we'll show that if her value $b_i<\nicefrac{2}{3}-\epsilon$ (which means her virtual value $\tilde{\varphi}_i(b_i)=\varphi_i(b_i)<\nicefrac{1}{3}-2\epsilon$) for some small $\epsilon$ she can only trade in the GSOM with exponentially small probability. Notice that buyer $i$ can only trade in GSOM with a seller whose cost is smaller than $\nicefrac{1}{3}-2\epsilon$. Let $\mathcal{S}$ be the number of sellers whose cost is smaller than $\nicefrac{1}{3}-2\epsilon$, and let $\mathcal{B}$ be the number of buyers whose value is at least $\nicefrac{2}{3}-\epsilon$. Then if buyer $i$ trades in GSOM, $\mathcal{B}\leq\mathcal{S}$. This is because there are at least $\mathcal{B}$ buyers with value larger than $i$ (and thus have a larger virtual value since since all buyers are i.i.d.). If $\mathcal{B}>\mathcal{S}$, those buyers will take away all the sellers with cost $<\nicefrac{1}{3}-2\epsilon$. It contradicts with the fact that buyer $i$ must trade with a seller with cost $<\nicefrac{1}{3}-2\epsilon$ in GSOM. Notice that again $\mathcal{S}$ and $\mathcal{B}$ are the sum of independent Bernoulli random variables. And the expectation of $\mathcal{S}$ $(\nicefrac{1}{3}-2\epsilon)n$ is smaller than the expectation of $\mathcal{B}$, $(\nicefrac{1}{3}+\epsilon)(n-1)$ when $\epsilon=n^{-\frac{1}{3}}$ and $n$ sufficiently large. By Chernoff bound (and a similar calculation as above), $\mathcal{B}\leq\mathcal{S}$ happens with exponentially small probability. The expected GFT contributed by each buyer $i$ is at most
$$\int_{\frac{2}{3}-n^{-\frac{1}{3}}}^1 b_idb_i+o(1)=\frac{5}{18}+o(1).$$

Similarly, for each seller $j$ whose value $s_j<\nicefrac{1}{3}-\epsilon$ for some small $\epsilon$, she will trade in GSOM with high probability. The expected GFT contributed by each seller (a negative term) is at most $-\nicefrac{1}{18}+o(1)$. Thus GSOM obtains GFT at most $\nicefrac{2n}{9}+o(n)$.

We conclude that the expected GFT the mechanism of \citet{BrustleCFM17} obtains is at most $\nicefrac{2n}{9}+o(n)$, which is only a constant fraction of the second-best mechanism.
\end{proof}

\begin{proof}[Proof of \cref{example:max-not-monotone}]
We first note that $\tilde{\varphi}_1(b_1)=\varphi_1(b_1)=2b_1-90$, that $\tilde{\varphi}_2(b_2)=\varphi_2(b_2)=2b_2-30$, that $\tilde{\tau}_1(s_1)=\tau_1(s_1)=s_1=0$, and that if $s_2=0$ then $\tilde{\tau}_2(s_2)=\tau_2(s_2)=0$ and else (i.e., if $s_2=25$, then) $\tilde{\tau}_2(s_2)=\tau_2(s_2)=25+\frac{25}{4}>30$.

\begin{itemize}
\item
We first analyze the case in which $s_2=25$.
\begin{itemize}
\item
If $b_2=24$, then the first-best matches at only one pair (since $b_2<s_2$), and so there is no trade in TR and the RVWM outcome is chosen. Since $b_2<s_2$ and $\tilde{\tau}_1(s_1)=0<b_2$ we have that buyer~$2$ trades in GBOM if and only if $b_1<b_2=24$, i.e., with probability $\nicefrac{24}{90}$.
Since $b_2<s_2$ and $\tilde{\varphi}_2(b_2)=\tilde{\varphi}_2(24)=18>s_1$, we have that buyer~$2$ trades in GSOM if and only if $\tilde{\varphi}_1(b_1)<\tilde{\varphi}_2(b_2)=18$, i.e., if and only if $b_1<54$, i.e., with probability $\nicefrac{54}{90}$.

So, the overall interim probability of trade of buyer~$2$ when $b_2=24$ and $s_2=25$ is $\nicefrac{39}{90}$.
\item
If $b_2=26$, then we consider a few cases.
\begin{itemize}
\item
If $b_1<25$, then the first-best matches only one pair (since $b_1<s_2$) and so there is no trade in TR and the RVWM outcome is chosen. Since $b_1<s_2$ and $\tilde{\tau}_1(s_1)=0<b_2$ and $b_1<b_2$, we have that buyer~$2$ trades in GBOM. Since $\tilde{\varphi}_2(b_2)=\tilde{\varphi}_2(26)=22>s_1$ and $\tilde{\varphi}_2(b_2)=22>-40\ge\tilde{\varphi}_1(b_1)$, we have that buyer~$2$ trades in GSOM. So, in this case buyer~$2$ trades with overall probability $1$.
\item
If $b_1\in[25,26]$, then the first-best matches all agents and so TR will trade buyer~$2$ and seller~$1$. As in the case where $b_1<25$, buyer~$2$ trades with probability $1$ in with GBOM and GSOM. So, in this case buyer~$2$ trades with overall probability $1$ as well.
\item
If $b_1>26$, then the first-best matches all agents and so TR will trade buyer~$1$ and seller~$1$ (so buyer~$2$ does not trade). Since $\tilde{\tau}_2(s_2)>30>b_2$, GBOM trades only one edge, and since $\tilde{\varphi}_2(b_2)=22<s_2$, GSOM trades at most one edge, so the GFT of TR is at least that of RVWM (with equality if and only if the TR and RVWM outcomes coincide), and so the TR outcome is chosen and buyer~$2$ does not trade. So, in this case buyer~$2$ does not trade.
\end{itemize}
So, the overall interim probability of trade of buyer~$2$ when $b_2=26$ and $s_2=25$ is $\nicefrac{26}{90}$.
\end{itemize}
\item
We now analyze the case in which $s_2=0$. In this case the first-best matches all agents, and moreover, GBOM has all agents trading. Furthermore, regardless of whether $b_2=24$ or $b_2=26$, since $\tilde{\varphi}_2(b_2)\ge0=s_1=s_2$, buyer~$2$ trades in GSOM. So, RVWM has buyer~$2$ trading with probability~$1$.
\begin{itemize}
\item
If $b_1<b_2$, then $b_2$ trades in the TR mechanism. So, buyer~$2$ trades with probability~$1$ in this case.
\item
If $b_2<b_1<45$, then $b_2$ does not trade in the TR mechanism. In this case, since $\tilde{\varphi}_1(b_1)<0=s_1=s_2$, only buyer~$2$ trades in GSOM. So, regardless of whether $b_2=24$ or $b_2=26$, we have that $GFT(TR)=b_1<\frac{b_1+b_2}{2}+\frac{b_2}{2}=GFT(GBOM)+GFT(GSOM)$, so the RVWM outcome is chosen and buyer~$2$ trades with probability~$1$.
\item
If $b_2<b_1$ and $b_1>45$, then since $\tilde{\varphi}_1(b_1)>0=s_1=s_2$, all agents trades in GSOM as well, regardless of whether $b_2=24$ or $b_2=26$. So, the RVWM outcome is efficient and is chosen, and buyer~$2$ trades with probability~$1$.
\end{itemize}
So, the overall interim probability of trade of buyer~$2$ when $s_2=0$, regardless of whether $b_2=24$ or $b_2=26$, is $1$.
\end{itemize}
Combining all of the above, we have that the overall interim probability of trade of buyer~$2$ with $b_2=24$, over the distributions of all other values and costs and over the randomness of the mechanism,  is $\nicefrac{4}{5}\cdot\nicefrac{39}{90}+\nicefrac{1}{5}\cdot1$ and the overall interim probability of trade of buyer~$2$ with $b_2=26$, over the distributions of all other values and costs and over the randomness of the mechanism,  is $\nicefrac{4}{5}\cdot\nicefrac{26}{90}+\nicefrac{1}{5}\cdot1$, which is strictly smaller.
\end{proof}

\begin{proof}[Proof of \cref{example:rvwm-switch-not-monotone}]
For the case in which $s_2=25$, the analysis coincides with that of \cref{example:max-not-monotone}. For the case in which $s_2=0$, the first-best matches all agents and therefore $q=2$ and TR is chosen. So, buyer~$2$ trades if and only if $b_2>b_1$. So, if $b_2=24$ then this occurs with probability $\nicefrac{24}{90}$ and if $b_2=26$ then this occurs with probability $\nicefrac{26}{90}$.

Combining all of the above, we have that the overall interim probability of trade of buyer~$2$ with $b_2=24$, over the distributions of all other values and costs and over the randomness of the mechanism,  is $\nicefrac{4}{5}\cdot\nicefrac{39}{90}+\nicefrac{1}{5}\cdot\nicefrac{24}{90}=\nicefrac{36}{90}$ and the overall interim probability of trade of buyer~$2$ with $b_2=26$, over the distributions of all other values and costs and over the randomness of the mechanism,  is $\nicefrac{4}{5}\cdot\nicefrac{26}{90}+\nicefrac{1}{5}\cdot\nicefrac{26}{90}=\nicefrac{26}{90}$, which is strictly smaller.
\end{proof}

\section{Proofs of Lemmas~\ref{ro} and~\ref{ro-trade-probability}}
\label{app:ro}

\begin{proof}[Proof of \cref{ro}]
We start by proving \cref{ro-properties}. 
We will show all of these properties for the SO mechanism with the given parameters. Similar arguments show them for the BO mechanism, and therefore for the RO mechanism. That the SO  mechanism with the given parameters is BIC for the seller is immediate since the the seller chooses a price that maximizes her expected utility over the distribution from which the buyer is drawn. That SO is ex-post IC and ex-post IR for the buyer is immediate from the buyer choosing whether to accept or reject the offer in a way that maximizes her utility. Strong budget balance is also immediate from definition. Finally, to show that that the SO  mechanism with the given parameters is ex-post IR for the seller, we note that since the seller's cost is drawn from $D_s$, and since $\bar{s}\ge\sup\supp D_s$, we have that $\bar{s}$ is at least the seller's cost. Therefore, seller $j$ can always ask for a price equal to his cost (without violating the constraint $\bar{s}$), which will result in zero utility for her (regardless of whether the buyer accepts or rejects this price), and in particular guarantees nonnegative utility for her.

We move on to prove \cref{ro-price-within-constraints}.
We will show that any offered price is at least $\bar{b}$. An analogous proof shows it to be at most $\bar{s}$. In case of the BO mechanism, this holds by definition since any offer by the buyer is constrained to be at least $\bar{b}$. We note that any offer by the seller will also be at least $\bar{b}$, since this seller knows that the buyer will buy at this price with probability $1$ since $\bar{b}\le\inf\supp D_b$ (and since this price is at least $\bar{s}$, offering it does not violate this seller's constraint), and therefore the seller will never make a lower offer, and the proof is complete.
\end{proof}

\begin{proof}[Proof of \cref{ro-trade-probability}]
We start by proving \cref{ro-trade-probability-always-trade}.
If $\bar{s}\le b$, then any price offered by the seller in the SO part of the RO mechanism with the given parameters will be accepted by the buyer (as it will be at most $\bar{s}$ and therefore at most $b$). Similarly, if $\bar{b}\ge s$, then any price offered by the buyer in the BO part of the RO mechanism with the given parameters will be accepted by the seller (as it will be at least $\bar{b}$ and therefore at least $s$). Either way, trade will occur in the RO mechanism with probability at least $\nicefrac{1}{2}$.

We move on to prove \cref{ro-trade-probability-restrictions-dont-matter}.
It is enough to show that if trade occurs in the SO part of the latter mechanism, then trade occurs also in the SO part of the former mechanism. (The BO part is handled analogously.)
If trade occurs in the SO part of the latter mechanism (unconstrained and unconditioned), then it means that in that mechanism, the price $p$ that maximizes the revenue of the seller from $D_b$ (unconstrained and unconditioned), which is the price that was offered, is at most $b$ (since the offered price is accepted) and at least $s$ (since the mechanism is ex-post IR). Therefore, since $b\le\bar{s}$ and $s\ge\bar{b}$, we have $\bar{s}\ge b\ge p \ge s \ge\bar{b}$, and so $p$ is also the price that maximizes the revenue of the seller from $D_b|_{\ge\bar{b}}$ constrained upon the price being at most $\bar{s}$, and so this is also (at most, in case of multiple utility-maximizing prices) the price offered in the SO mechanism with parameters $\bar{s}$ and $D_b|_{\ge\bar{b}}$, and so the price offered by the seller is accepted in this mechanism as well.
\end{proof}

\section{The Trade Reduction Mechanism for\texorpdfstring{\\}{ }Matching Markets: Proofs}
\label{app:matching-tr}
		
\begin{proof}[Proof of \cref{thm:matching-TR}]
	We first observe that the allocation is indeed feasible, that is, that we can perfectly match all winning buyers and sellers. Indeed, the set of winners can be obtained by taking the matching  $M(\Bb,\Bs)$ and then removing the agents that correspond to one edge between any two classes $t$ and~$t'$ that are trading (have  $r_{t,t'}>0$), and then switching agents of the same class (removing every agent in the leftover matching that has value lower than a removed agent of the same class, and adding the removed agent in her stead), maintaining a perfect matching of $\mathit{TR}(\Bb,\Bs)$. Thus, there is a matching of the winners $\mathit{TR}(\Bb,\Bs)$ such that there are exactly $r_{t,t'}-1$ trades of agents of classes~$t$ and $t'$ whenever $r_{t,t'}>0$. The reduced agents can be perfectly matched with exactly a single edge between agents  of classes $t$ and $t'$ whenever $r_{t,t'}>0$.

Now, the \lcnamecref{thm:matching-TR} will directly follow from the following sequence of \lcnamecrefs{matching-TR-incentive}.
\begin{claim}\label[claim]{matching-TR-incentive}
	The Trade Reduction Mechanism for matching markets is ex-post IR and ex-post IC.
\end{claim}
\begin{proof}
We first observe that the TR mechanism for matching markets is monotone. It is enough to show this for the buyers.
Assume that the value of a winning buyer $i$ increases by $\delta$. We show that she still wins after the increase. Every matching that includes this buyer improves by the same amount $\delta$, while the value of any other matching did not change, so the same matching $M(\Bb,\Bs)$ will be picked after the value increase (ties are broken the same way, independent of values). Finally, the reduction will also not change as for any class $t$, the values of $q_t$ and $d_t$ did not change, and if buyer $i$ is of class $t$, she will still not be in the set of $d_t$ lowest-value buyers after her value has increased, so she will not be reduced.

To complete the proof that the mechanism is ex-post IR and ex-post IC, we need to show that payments are by critical values.
Indeed, assume that for a winning buyer $i$ of class $t$, the value of the highest reduced buyer of class $t$ (if $i$ bids truthfully) is $x$. If buyer $i$ of class $t$ changes her bid but keeps it  above $x$, she wins. Now assume that she drops her bid below $x$. If she is not in $M(\Bb,\Bs)$ she loses. If she is in $M(\Bb,\Bs)$ and bids below $x$, then the matching $M(\Bb,\Bs)$ will contain the exact same set of agents, and will be the same (due to the tie breaking rule which is independent of the actual bids), so she will lose while the agent with bid $x$ will win instead of her, as both $q_t$ and $d_t$ are the same but now $i$ is no longer one of the $q_t-d_t$ highest-bidding agents of class $t$.
Thus, bidding above $x$ implies winning for $i$, while bidding below $x$ implies losing, so $x$ is indeed the critical value for buyer $i$ to win. Similar arguments prove the \lcnamecref{matching-TR-incentive} for sellers.
\end{proof}

\begin{claim}\label[claim]{matching-TR-budget}
	The Trade Reduction Mechanism for matching markets is ex-post (direct trade) weakly budget balanced.
\end{claim}	
\begin{proof}
To see that the mechanism is ex-post (direct trade) weakly BB we prove that every trade is ex-post weakly BB.
Indeed, consider a trade between two agents of respective classes $t$ and $t'$ in the matching of $\mathit{TR}(\Bb,\Bs)$ that contains exactly $r_{t,t'}-1$ trades of agents of classes $t$ and $t'$ whenever $r_{t,t'}>0$, and consider the reduced edge between these two classes. The buyer pays at least the value of the reduced buyer of the same class on the reduced edge between classes $t$ and $t'$, while the seller receives at most the cost of the reduced seller on that edge.
As that reduced edge has non-negative gain (otherwise removing it will increase the welfare of the first-best matching $M(\Bb,\Bs)$), the trade is ex-post weakly BB.
\end{proof}

\begin{claim}\label[claim]{matching-TR-gains}
	For any profile $(\Bb,\Bs)$, the fraction of the realized gains from trade (first-best) that the Trade Reduction Mechanism for matching markets obtains ex-post is at least $\min\bigl\{1- \frac{d_t}{q_t} ~\big|~ \text{class $t$  s.t.\ $q_t>0$}\bigr\}$.
\end{claim}	

\begin{proof}
	Recall that $T_t$ is the set of class-$t$ agents.
	Let $\alpha(\Bb,\Bs) = \min\bigl\{1- \frac{d_t}{q_t} ~\big|~ \text{class $t$  s.t.\ $q_t>0$}\bigr\}$. To prove the claim that the mechanism guarantees an $\alpha$ fraction of the welfare of $OPT(\Bb,\Bs)$, we let $v_k$ be the value of agent $k$ ($v_k=b_i$ for a buyer $k=i$, and $v_k=-s_j$ for seller $k=j$), and assuming $W$ is the set of agents in $M(\Bb,\Bs)$ we
	observe that
	$$OPT(\Bb,\Bs)=\sum_{(i,j)\in M(\Bb,\Bs)} (b_i-s_j) =  \sum_{t:q_t>0} \ \ \ \ \sum_{k\in W\cap T_t} v_k$$.

	As for each class $t$ with $q_t>0$ we remove $d_t$ agents each with value at most the value of any winner, we obtain at least a $\frac{q_t-d_t}{q_t}$ fraction of the value of agent of class $t$. Thus,
	$$\alpha(\Bb,\Bs) \cdot OPT(\Bb,\Bs)\leq \sum_{t: q_t>0} \frac{q_t-d_t}{q_t} \ \ \ \ \sum_{k\in W\cap T_t} v_k\le TR(\Bb,\Bs),$$
	as needed.
\end{proof}

\cref{thm:matching-TR} follows from \cref{matching-TR-incentive,matching-TR-budget,matching-TR-gains}.
\end{proof}

\begin{proof}[Proof of \cref{cor:matching-TR-beta}]
Let $\beta(\Bb,\Bs) = \min\bigl\{ 1- \frac{1}{r_{t,t'}} ~\big|~\text{classes $(t,t')$ s.t.\ $ r_{t,t'}>0$}\bigr\}$.
By the guarantee of \cref{thm:matching-TR} with respect to $\alpha(\Bb,\Bs)$, it is enough to show that $\beta(\Bb,\Bs)\le\alpha(\Bb,\Bs)$. Indeed, for every class $t$ such that $q_t>0$, since $q_t=\sum_{t'} r_{t,t'}$:
\[\frac{q_t-d_t}{q_t}=\sum_{t':r_{t,t'}>0}\frac{r_{t,t'}-1}{q_t}=\sum_{t':r_{t,t'}>0}\frac{r_{t,t'}-1}{r_{t,t'}}\cdot\frac{r_{t,t'}}{q_t}\ge\min_{t':r_{t,t'}>0}\frac{r_{t,t'}-1}{r_{t,t'}},\]
where the inequality is since a weighted average of values is always at least the minimal value. Taking the minimum of both sides of the obtained  inequality over all classes $t$ s.t.\ $q_t>0$, we obtain that $\alpha(\Bb,\Bs)\ge\beta(\Bb,\Bs)$, as required.
\end{proof}

\section{Additional Preliminaries for Appendices~\ref{app:matching-offering} and~\ref{app:matching-hybrid}}
\label{sec:matching-preliminaries}

\subsection{Notation}\label{app:notation}

First we give some notations specialized in this setting.
Given profile $(\Bb,\Bs)$, let $M(\Bb,\Bs)$ be the first-best matching, or the maximum weight matching\footnote{We break ties lexicographically by IDs.}, under graph $G$ with edge weight $b_i-s_j$ on each edge $(i,j)\in E$. For each agent $a$, denote $M_{-a}(\Bb,\Bs)$ the maximum weight matching\footnote{Follow the same breaking tie rules as the first-best matching.} after removing $a$ and its related edges. For each buyer $i$ such that $(i,j)\in M(\Bb,\Bs)$, let $P_i(\Bb,\Bs)$ be the VCG payment of buyer $i$. Formally,
$$P_i(\Bb,\Bs)= \sum_{(i',j')\in M_{-i}(\Bb,\Bs)}(b_{i'}-s_{j'})- \sum_{(i',j')\in M(\Bb,\Bs)}(b_{i'}-s_{j'})+b_i$$

Similarly, let $P_j(\Bb,\Bs)$ be the VCG payment received by seller $j$:
$$P_j(\Bb,\Bs)= \sum_{(i',j')\in M(\Bb,\Bs)}(b_{i'}-s_{j'})- \sum_{(i',j')\in M_{-j}(\Bb,\Bs)}(b_{i'}-s_{j'})+s_j$$

For simplicity, when the valuation profile $(\Bb,\Bs)$ is fixed, we will abuse the notation and use $M$ (or $M_{-a}$, $P_i$, $P_j$) instead in the proof, without writing the valuation profile.

\subsection{Lexicographic Tie-Breaking by ID}\label{app:ties}

In this section, we define the tie-breaking rule that we use whenever we have to choose between multiple maximum weight matchings when picking a matching with maximum weight anywhere throughout this paper.
We first define a strict total order over matchings, which we call the \emph{Lexicographic order by IDs}.

\begin{definition}[Lexicographic order by IDs]\label{def:lex-order}
Fix a bipartite graph,
and let $M'$ and $M''$ be two matchings in this graph.
The \emph{Lexicographic order by IDs} decides which of $M'$ and $M''$ is ranked higher as follows.
It first sorts the edges of the matching by the index of the buyer.
For each $k$, let $(i'_k,j'_k)$ and $(i^{''}_k,j^{''}_k)$ be the $k$\textsuperscript{th} sorted edges (according to the index of the buyer) in $M'$ and in $M''$, respectively.
Let $k$ be the lowest index such that it is not the case that the two edges $(i'_k,j'_k)$ and $(i^{''}_k,j^{''}_k)$ are both defined and are the same edge.
\begin{itemize}
	\item If one matching has a $k$\textsuperscript{th} edge while the other does not, then
	the matching with more edges is ranked higher.
	\item Otherwise,
	the matching with the lower buyer index in the $k$\textsuperscript{th} edge is ranked higher.
	\item Otherwise,
	the matching with the lower seller index in the $k$\textsuperscript{th} edge is ranked higher.
\end{itemize}
\end{definition}

As noted above, throughout this paper when two matchings have the same weight, we use the Lexicographic by IDs order to break ties when choosing a maximum weight matching, so we in fact choose the lexicographically-by-IDs-highest matching among those with maximum weight. We will refer to this practice as using the \emph{Lexicographic by IDs tie-breaking rule}.
We will now formalize the two properties of this tie-breaking rule, which we will use in our analysis:\footnote{Indeed, our results would still hold for any other tie-breaking rule that satisfies these two properties.}

\begin{lemma}
	The Lexicographic by IDs tie-breaking rule satisfies the following two properties:
	\begin{itemize}
		\item The tie-breaking is \emph{weight independent}: ties between maximum weight matchings are broken independently of any weight function. That is, if $W$ and $W'$ are two weight functions,
		if $\mathcal{M}$ and $\mathcal{M}'$ are the respective corresponding sets of maximum weight matchings, and if $M$ and $M'$ are the respective corresponding chosen
		matchings, then if $\mathcal{M}\subseteq\mathcal{M'}$ and $M'\in\mathcal{M}$, then $M=M'$.
		So, the set of matched nodes that this tie-breaking rule picks (among all possible maximum weight options), as well as the matching that this rule picks within that set, does not depend on the weight function.
		\item The choice function is \emph{subset consistent}: if the chosen maximum weight matching among all matchings of the vertices $U$ is the matching $M$, then for any $(i,j)\in M$, the chosen maximum weight matching among all matchings of the vertices $U\setminus\{i,j\}$ is the matching $M\setminus\{(i,j)\}$.
	\end{itemize}
\end{lemma}

\begin{proof}
Weight-independence is by definition of the Lexicographic order by IDs. For subset consistency, let $M'\ne M\setminus\{(i,j)\}$ be another maximum weight matching of $U\setminus\{(i,j)\}$, and note that  when adding the edge $(i,j)$ to $M'$, one obtains a maximum weight matching of $U$. By definition of $M$, it is ranked higher than $M'\cup\{(i,j)\}$ by the Lexicographic by IDs order, and since the shared edge makes no difference in the tie breaking, we have that after its removal $M\setminus\{(i,j)\}$ is (still) ranked higher than $M'$ by the Lexicographic by IDs order (so the tie is be broken in the same way).
\end{proof}

\section{The Offering Mechanism for\texorpdfstring{\\}{ }Matching Markets: Proofs}
\label{app:matching-offering}

We will now prove \cref{offering}, which states that the Offering Mechanism for matching markets is BIC, ex-post IR, ex-post (direct trade) strongly budget balanced, and ex-ante guarantees at least a $\nicefrac{1}{4}$-fraction of the optimal GFT (second-best). The Offering Mechanism is ex-post strongly (direct trade) budget balanced as the RO mechanism is ex-post (direct trade) strongly budget balanced. To show the remaining properties, we first develop some machinery.

\subsection{Supporting Machinery}

\begin{lemma}\label[lemma]{offering-constraint-vcg}
In the offering mechanism when run on a profile $(\Bb,\Bs)$, for every $(i,j)\in M(\Bb,\Bs)$ the following hold:
\begin{itemize}
\item
If $i\in M_{-j}(\Bb,\Bs)$, then $\bar{s}=P_j(\Bb,\Bs)$.
\item
If $j\in M_{-i}(\Bb,\Bs)$, then $\bar{b}=P_i(\Bb,\Bs)$.
\end{itemize}
\end{lemma}

\begin{proof}
We will prove the first statement (the second is analogous).
Since $i\in M_{-j}(\Bb,\Bs)$, the VCG price of buyer $i$ in the market without seller $j$ is the minimal bid that causes her to be in the first-best in that market, and so $\bar{s}=P_i(\Bb,\Bs_{-j})$.
Now observe that:
\begin{multline*}
P_j(\Bb,\Bs)= \sum_{(i',j')\in M}(b_{i'}-s_{j'})- \sum_{(i',j')\in M_{-j}}(b_{i'}-s_{j'})+s_j=\\
\sum_{(i',j')\in M\setminus\{(i,j)\}}(b_{i'}-s_{j'})- \sum_{(i',j')\in M_{-j}}(b_{i'}-s_{j'})+b_i
=P_i(\Bb,\Bs_{-j})=\bar{s}\tag*{\qedhere}
\end{multline*}
\end{proof}

\noindent
We are now ready to prove the first two parts of \cref{offering-good-parameters}:

\begin{claim}\label{offering-ir}
For every $(i,j)\in M(\Bb,\Bs)$, it holds that $\bar{s}\ge s_j$ and $\bar{b}\le b_i$.
\end{claim}

\begin{proof}
We will show the former; the latter is analogous. If $\bar{s}=\infty$ the the claim immediately holds, so we assume that $\bar{s}<\infty$. Consider the profile $((\Bb_{-i},b'_i),\Bs)$ for $b'_i=\max\{\bar{s}+1,b_i\}$. Since we have only increased the bid of $i$, we still have that $(i,j)\in M((\Bb_{-i},b'_i),\Bs)$. By definition, $\bar{s}$ is the same for the profile $((\Bb_{-i},b'_i),\Bs)$ as it is for $(\Bb,\Bs)$. By definition of $\bar{s}$, we have by $b'_i>\bar{s}$ that $i\in M_{-j}((\Bb_{-i},b'_i),\Bs)$. Therefore, by \cref{offering-constraint-vcg}, $\bar{s}=P_j((\Bb_{-i},b'_i),\Bs)$. Since $(i,j)\in M((\Bb_{-i},b'_i),\Bs)$, we have that $s_j\le P_j((\Bb_{-i},b'_i),\Bs)$, and so $s_j\le \bar{s}$, as required.
\end{proof}

\begin{claim}\label{vcg-budget-deficit}
For every $(\Bb,\Bs)$ and $(i,j)\in M(\Bb,\Bs)$, it is the case that $P_j(\Bb,\Bs)\ge P_i(\Bb,\Bs)$.
\end{claim}

\begin{proof}
by the Second Welfare Theorem, there exist prices $p=(p_{j'})_{j'\in S}$ (where $p_{j'}$ denotes a price for the good of seller $j$) such that $(M(\Bb,\Bs);p)$ is a Walrasian equilibrium. Therefore, $p_j$ is a price received by $s_j$ and paid by $b_i$ in some Walrasian equilibrium. By Theorem 8 of \citet{GulS99}, we therefore have that $P_j(\Bb,\Bs)\ge p_j \ge P_i(\Bb,\Bs)$, completing the proof.
\end{proof}

\noindent
We are now ready to prove the third and final part of \cref{offering-good-parameters}:

\begin{claim}\label{offering-constraint-order}
For every $(i,j)\in M(\Bb,\Bs)$, it holds that $\bar{s}\ge\bar{b}$.
\end{claim}

\begin{proof}
Assume for contradiction that there exists a profile $(\Bb,\Bs)$ and a pair $(i,j)\in M(\Bb,\Bs)$ such that $\bar{s}<\bar{b}$. By \cref{offering-ir} we have that $b_i\ge\bar{b}>\bar{s}$ and $s_j\le\bar{s}<\bar{b}$. Therefore, we have by definition of $\bar{s}$ and $\bar{b}$ that both $i\in M_{-j}(\Bb,\Bs)$ and $j\in M_{-i}(\Bb,\Bs)$. Therefore, by \cref{offering-constraint-vcg}, $\bar{s}=P_j(\Bb,\Bs)$ and $\bar{b}=P_i(\Bb,\Bs)$, and so $P_j(\Bb,\Bs)<P_i(\Bb,\Bs)$ --- a contradiction to \cref{vcg-budget-deficit}.
\end{proof}

\begin{proof}[Proof of \cref{offering-good-parameters}]
Follows from \cref{offering-ir,offering-constraint-order}.
\end{proof}

\begin{lemma}\label[lemma]{lem:VCG-price}
Fix valuation profile $(\Bb,\Bs)$ and a pair $(i,j)\in M$. If $j\not\in M_{-i}$, then $P_i(\Bb,\Bs)=s_j$. Similarly if $i\not\in M_{-j}$, then $P_j(\Bb,\Bs)=b_i$.
\end{lemma}
\begin{proof}
We only give the proof for $P_i = P_i(\Bb,\Bs)$ and similar argument holds for seller's VCG payment $P_j = P_j(\Bb,\Bs)$.
If $j\not\in M_{-i}$, it holds that $M_{-i}=M\setminus\{(i,j)\}$ by subset consistency of the tie breaking rule. Now, by definition of $P_i$, we have $P_i=s_j$.
\end{proof}

\subsection{Incentive Guarantees}

\begin{claim}
The Offering Mechanism is ex-post IR.
\end{claim}

\begin{proof}
That the Offering Mechanism is ex-post IR follows from \cref{offering-ir,ro}.
\end{proof}

\begin{claim}\label{offering-bic}
The Offering Mechanism is BIC.
\end{claim}

\begin{proof}
We will prove that the Offering Mechanism is BIC for the seller. A similar argument holds for the buyer.
For each seller $j$ with cost $s_j$, suppose she misreports her cost to be $s_j'\not=s_j$. We will show that taking expectation over other agents' valuation profile $\Bb,s_{-j}$, the expected utility of $s_j$ when reporting truthfully is at least the expected utility of seller $j$ with true cost $s_j$ when reporting $s_j'$.

We first consider $\Bb,s_{-j}$ such that $j$ is not in the first-best $M(\Bb,s_{-j},s_j)$. It is sufficient to consider manipulations $s_j'$ that cause $j$ to become part of the first-best $M(\Bb,s_{-j},s_j')$. Let $s_j'$ be such a manipulation, and note that in this case, $s_j\ge P_j(\Bb,s_{-j},s_j')\ge s_j'$ (since $P_j(\Bb,s_{-j},s_j')$ is the threshold bid of seller $j$ to become part of the first-best). Let $i$ be the agent such that $(i,j)\in M(\Bb,s_{-j},s_j')$. We will complete the proof of this case by considering two cases. First, if $i\in M_{-j}(\Bb,s_{-j},s_j')$, then $\bar{s}=P_j(\Bb,s_{-j},s_j')$ by \cref{offering-constraint-vcg}. Therefore, by \crefpart{ro}{price-within-constraints}, seller $j$ with reported cost $s_j'$ can only trade with $i$ in the RO mechanism at a price $p\leq\bar{s}=P_j(\Bb,s_{-j},s_j')\le s_j$, which derives non-positive utility for seller $j$. Second, if $i\notin M_j(\Bb,s_{-j},s_j')$, then by \cref{lem:VCG-price}, $P_j(\Bb,s_{-j},s_j')=b_i$. Since the mechanism is ex-post IR for buyer $i$, seller $j$ can only trade with $i$ (who has value $b_i$) in the RO mechanism at a price at most $b_i=P_j(\Bb,s_{-j},s_j)\le s_j$, which again derives non-positive utility for seller $j$.

Now for every buyer~$i$, consider those $\Bb,s_{-j}$ such that $(i,j)\in M(\Bb,s_{-j},s_j)$.
In this case, we note that if seller~$j$ misreports to $s'_j$, then either the first-best is unchanged (and so $j$ participates in the same RO mechanism with the same buyer $i$) or seller~$j$ is no longer in the first-best, receiving utility $0$. Either way, she cannot change the RO mechanism that is run, or the buyer that she is facing.
When the BO mechanism is processed, if seller $j$ is in the first-best, she will be asked to accept a price. This is ex-post truthful.

When the SO mechanism is processed, then it is enough to show that for every $i$ and fixed $b_{-i},s_{-j},s_j$, in expectation over all $b_i$ such that $(i,j)\in M(b_{-i},b_i,s_{-j},s_j)$, the utility of $j$ when she reports $s_j'\not=s_j$ (denoted as $u_j(s_j')$) is at most the utility of $j$ when she reports $s_j$ (denoted as $u_j(s_j)$) truthfully. Notice that either $i$ can never connect to $j$ in the first-best $M(b_{-i},b_i,s_{-j},s_j)$, or $P_i(\Bb,s_{-j},s_j)$ (which does not depend on $b_i$) is the threshold bid of buyer $i$ to connect to $j$ in the first-best. Thus it is enough to prove the claim that when there exists a bid for $i$ such that $i$ connects to $j$ in the first-best, in expectation over $b_i\geq P_i(\Bb,s_{-j},s_j)$, the utility of $j$ when she reports $s_j'$ is at most the utility of $j$ when she reports~$s_j$.

Note that fixed $b_{-i},s_{-j},s_j$, when $b_i\geq P_i(\Bb,s_{-j},s_j)$, the outcome of the Offering Mechanism for~$j$ is as if the SO mechanism with parameters $\bar{s}$ and $D^B_i|_{\ge\bar{b}}$ had been run between $j$ and $i$. Notice that the parameters $\bar{s}$ and $\bar{b}$ do not depend on $s_j$ or on $b_i$.
By \crefpart{ro}{properties}, the SO mechanism with parameters~$\bar{s}$ and $D^B_i|_{\ge\bar{b}}$ is BIC when the buyer's valuation is drawn from $D^B_i|_{\ge\bar{b}}$. In other words, 
\begin{equation}\label{equ:utility}
\E_{b_i\geq \bar{b}}[u_j(s_j')]\leq \E_{b_i\geq \bar{b}}[u_j(s_j)].
\end{equation}
Note that $P_i(\Bb,s_{-j},s_j)\ge\bar{b}$ by \cref{offering-ir}. If $P_i(\Bb,s_{-j},s_j)=\bar{b}$, then the above claim trivially holds.
Otherwise we have to reason about the case $P_i(\Bb,b_{-j},s_j)>b_i\ge\bar{b}$, which is included in the expectation in \cref{equ:utility} but not in the expectation in the above claim.

In this case, since $P_i(\Bb,s_{-j},s_j)>\bar{b}$,
then \cref{offering-constraint-vcg,lem:VCG-price}, $P_i(\Bb,s_{-j},s_j)=s_j$. We notice that when $s_j>b_i\ge\bar{b}$, seller $j$ won't trade with buyer $i$ in the Offering Mechanism when $j$ reports $s_j$, as $(i,j)$ can't be in the first-best $M(\Bb,\Bs)$. The utility of $j$ (contributed by buyer $i$) is thus $0$ in this case. When $s_j>b_i\ge\bar{b}$ and $j$ reports $s_j'$, since the payment goes directly from buyer $i$ to seller $j$ when they trade, they do so at price at most $b_j<s_j$ (since we assume that buyer $i$ reports truthfully, and since the mechanism is ex-post IR for her), so seller $j$'s utility (contributed by buyer $i$) in this case is negative if they trade, and 0 otherwise, so it is nonpositive. Combined this with \cref{equ:utility}, we obtain
\[
\E_{b_i\geq s_j}[u_j(s_j')]\leq \E_{b_i\geq s_j}[u_j(s_j)]
\]
which finishes the proof as $s_j=P_i(\Bb,s_{-j},s_j)$.
\end{proof}

\subsection{Efficiency Guarantee}\label{matching-ex-ante}

Given a bipartite graph $(B,S,E)$ and two matchings $M$ and $M'$ over the graph,
a path is called an \emph{alternating path} of $M\cup M'$  if the edges on the path alternate between edges of $M$ and $M'$.
If $a_K=a_1$, we call it an \emph{alternating cycle}. A path is \emph{maximal} if it is not a sub-path of any other path.

It is well-known that the union of two matchings in a bipartite graph can be divided into disjoint maximal alternating paths and cycles.

\begin{observation}\label{obs:division-alternating}
	Given any set of nodes $V$ and two undirected graphs $G_1= (V,E_1)$ and $G_2 = (V,E_2)$ such that in both graphs the degree of any node is at most $1$ (i.e., each is a matching), it holds that in $G_{1,2}=(V, E_1\cup E_2)$ every node has degree at most $2$ and thus $G_{1,2}$ is a disjoint union of maximal alternating paths and maximal alternating cycles.
\end{observation}

Given a bipartite graph $(V_1,V_2,E)$, a set $U\subseteq V_1\cup V_2$ of nodes is \emph{matchable} if it is possible to find a perfect matching of all of the nodes in $U$ using edges in $E$. Note that if $U$ is matchable then $|U\cap V_1|= |U\cap V_2|$. A \emph{node weight function} is a function $W$ that assigns a weight $W(i)$ to any node $i\in V_1 \cup V_2$.
A \emph{node-based} weighted matching problem is a matching problem in which for some node weight function $W$,
the weight of every edge $(i,j)\in E$ is the sum of the weights of the two nodes incident on the edge, that is, $W(i,j)=W(i)+W(j)$.
The weight of a matchable set of nodes $U$ is $W(U) = \sum_{u\in U} W(u)$.
For a weighted matching problem, a \emph{weight-maximizing set} is a matchable set of nodes that has maximum weight, over all matchable sets.
Clearly, for any node-based  weighted matching problem, the weight of any matching over the same matchable set of nodes $U$ is the same.
Moreover, if $U$ is a weight maximizing set, then any perfect matching of it does not include any edge of negative weight.
For our mechanisms to work, we will need to carefully define the tie-breaking rule that will be used to choose the weight-maximizing set, as well as the perfect matching of its elements.

\begin{observation}\label[observation]{obs:same-matching}
	Fix a bipartite graph.
	Let $W_V$ and $W'_V$ be two node-based weight functions for the graph, and let $M$ and $M'$ be the two maximum weight matchings picked by the tie-breaking rule for these two weight functions, respectively.
	If the sets of matched nodes of $M$ and $M'$ are the same, then $M$ and $M'$ must be the same.
\end{observation}

We also observe that if an agent is in the first best, by changing his bid he cannot influence the picked matching while staying in the first best.
\begin{observation}\label[observation]{obs:first-best}
	Fix a bipartite graph.
	Assume that with node-based weight function $W$, the maximum weight matching $M$ is picked by the tie-breaking rule.
	Fix any $i$ and let $W'$ be a node-based weight function such that $W'(k)= W(k)$ for any $k\neq i$.  Let $M'$ be the the maximum weight matching picked for $W'$.
	Then if $i\in M'$ it holds that $M=M'$.
\end{observation}

We prove the following two lemmas about VCG prices which are both useful in our proofs.

\begin{corollary}\label{cor:vcg}
	Consider the VCG mechanism with lexicographic by IDs tie-breaking rule.
	
	If $(i,j)\in M(\Bb,\Bs)$ for some $(\Bb,\Bs)$ then for any $b'_i$ such that $i$ trades when the bids are $((\Bb_{-i}, b'_i),\Bs)$, it holds that buyer $i$ trades with $j$ and pays $P_i(\Bb,\Bs)$. Moreover, for any such $b'_i$ it holds that $b'_i\geq P_i(\Bb,\Bs)\geq s_j$.
	
	Similarly, if $(i,j)\in M(\Bb,\Bs)$ for some $(\Bb,\Bs)$ then for any $s'_j$ such that seller $j$ trades when the bids are $(\Bb,(\Bs_{-j}, s_j))$, it holds that $j$ trades with $i$ and pays $P_j(\Bb,\Bs)$. Moreover, for any such $s'_j$ it holds that $s'_j\leq P_j(\Bb,\Bs)\leq b_i$.
\end{corollary}
\begin{proof}
The inequality
$b'_i\geq P_i(\Bb,\Bs)$ holds by VCG being ex-post IR. It holds that $P_i(\Bb,\Bs)\geq s_j$ as otherwise, if
$P_i(\Bb,\Bs)<s_j$ then for $b_i$ s.t. $P_i(\Bb,\Bs)< b_i < s_j$ there is an inefficient trade in $M$, a contradiction. Similar arguments prove imply the claim for  seller $j$.
\end{proof}

Observe that $M=M(\Bb,\Bs)$, $M_1^*=M_1^*(\Bb,\Bs)$ and $M_2^*=M_2^*(\Bb,\Bs)$ are each a maximum weighted matching for some node based weight function, all defined over the same undirected bipartite graph $G=(S,B,E)$ and chosen using the same tie-breaking rule.
$M(\Bb,\Bs)$ is derived from the node-based function $W$ that assigns weight $b_i$ to any node $i\in B$ and weight $-s_j$ to any node $j\in S$.
Similarly, $M_1^*(\Bb,\Bs)$ is derived from the node-based function $W_1$ that assigns weight $\tilde{\varphi}_i(b_i)$ to any node $i\in B$ and weight $-s_j$ to any node $j\in S$, where $\tilde{\varphi}_i(b_i)$ is the ironed virtual value of $i$ when his value is $b_i$, as defined in \cref{sec:preliminaries-rvwm}. Finally, $M_2^*(\Bb,\Bs)$ is derived from the node-based function $W_2$ that assigns weight $b_i$ to any node $i\in B$ and weight $-\tilde{\tau}_j(s_j)$ to any node $j\in S$, where $\tilde{\tau}_j(s_j)$ is the ironed virtual cost of $j$ when his cost is $s_j$, as defined in \cref{sec:preliminaries-rvwm}.

A direct corollary of Observation \ref{obs:same-matching}, is that any alternating cycle in
$M\cup M_1^*$
cannot include more than two distinct nodes, as any such alternating cycle is actually two different matchings over the same matchable set of nodes.
We state the claim for  $M_1^*$; the same claim holds also for $M_2^*$.
\begin{corollary}\label[corollary]{cor:cycles}
	Let $(a_1a_2...a_K)$, be an alternating cycle of $M\cup M_1^*$. Then $K=3$. In other words, $a_1=a_3$ and the undirected edge
	$(a_1,a_2)$ is in both $M$ and in $M_1^*$.
\end{corollary}

For a bipartite graph $(V_1,V_2,E)$, we say that the node-based weight function $W$ is a \emph{$V_1$-weak-improvement of $W'$} if for any node $i\in V_1$ it holds that $W(i)\geq W'(i)$ and for any $j\in V_2$ it holds that $W(j)=W'(j)$.
By definition of $\tilde{\varphi}$ and $\tilde{\tau}$, we have that the node-based weight function $W$ used to derive $M$ is a $B$-weak improvement to the node-based weight function $W_1$ used to derive $M_1^*$ (and similarly,
$W$ is an $S$-weak improvement to the node-based weight function $W_2$ used to derive $M_2^*$).

\begin{lemma}\label[lemma]{lem:path}
	Fix a bipartite graph $(B,S,E)$, and assume that the node-based weight function $W$ is a $B$-weak-improvement of $W'$.
	Let $M$ and $M'$ be the maximum weight matchings that are picked by the
	tie-breaking rule for $W$ and $W'$ respectively.
	Consider a  maximal alternating path of $M\cup M'$ that is not a cycle.
	It holds that path cannot both start and end with an edge from $M'$.
	Moreover, it holds that the path (or its inverse) starts with a node in $B$ and that the first edge belongs to $M$.
\end{lemma}
\begin{proof}
	Let $A=(a_1a_2...a_K)$ be a maximal alternating path in $M\cup M'$ that is not a cycle, and let $U$ be the set of nodes in the path $A$.
	W.l.o.g., if there is a a node in $B$ on any end of the path, it is the first in the path.
	Assume that the path starts and ends with an edge from $M'$.
	In this case the path must have an odd number of edges (as any edge from $M'$ is followed by an edge from $M$), and it starts with a node $a_1\in B$ and ends with a node $a_k\in S$. If $W'(a_1)+ W'(a_k)<0$, then $U'=U\setminus \{a_1,a_k\}$ is matchable and has higher weight than $U$ for $W'$, a contradiction to the maximality of $M'$.
	If $W'(a_1)+ W'(a_k)\geq 0$ then since $W$ is a $V_1$-weak-improvement of $W'$ it holds that $W(a_1)+ W(a_k)\geq W'(a_1)+ W'(a_k)\geq 0$. If $W(a_1)+ W(a_k)>0$ then the matchable set $U'$ has higher weight than $U$ with respect to $W$, contradicting the maximality of $M$. If on the other hand $W(a_1)+ W(a_k)= W'(a_1)+ W'(a_k)=0$
	both $U$ and $U'$ are matchable sets of the same weight with respect to both $W$ and $W'$, so by set consistency of the tie breaking, both $M$ and $M'$ must have matched the same set, a contradiction.
	
	Now, if the path starts and ends with an edge in $M$, it has an odd number of edges, so it has a node in $B$ on one end and a node in $S$ on the other, and as we can assume w.l.o.g.\ that the node in $B$ is first, this completes the proof.
	
    We are left with the case that the path has an edge from $M$ on one end, and an edge from $M'$ on the other.
    In this case it has  an even number of edges and thus either both ends are in $B$, or both are in $S$.
    We prove that both are in $B$, completing the proof of the claim.
	Assume by way of contradiction that both $a_1$ and $a_k$ are in $S$.
	Since $W$ is a $B$-weak-improvement of $W'$, for any node $j\in S$ we have $W(j)=W'(j)$ and thus $W(a_1)=W'(a_1)$ and $W(a_k)=W'(a_k)$.
	If $W(a_1)<W(a_k)$ then $U\setminus \{a_k\}$ is a matchable set with higher weight than the set $U\setminus \{a_1\}$ with respect to $W$, contradicting the optimality of $M$. Similarly if   $W'(a_1)=W(a_1)>W(a_k)=W'(a_k)$ then $U\setminus \{a_1\}$ is a matchable set with higher weight than the set $U\setminus \{a_k\}$ with respect to $W'$, contradicting the optimality of $M'$. Thus it must be the case that $W'(a_1)=W'(a_k)$. So both matchable sets $U\setminus \{a_1\}$ and $U\setminus \{a_k\}$ have exactly the same weight with respect to both $W$ and $W'$, and as the tie breaking is
 weight-independent,	
 both should have picked the same set, a contradiction.
\end{proof}

From Lemma \ref{lem:path} and Corollary \ref{cor:cycles} we immediately get the following corollary.
\begin{corollary}\label[corollary]{cor:buyer-first}
	Let $(a_1a_2...a_K)$, be a maximal alternating path of $M\cup M_1^*$. Precisely one of the following holds:
	\begin{itemize}
	\item $K=3$, $a_1=a_3$ and the undirected edge $(a_1,a_2)\in M\cap M_1^*$,\quad or
	\item
	(w.l.o.g.)\ the path starts with a buyer and an edge from $M$.
	\end{itemize}
\end{corollary}	

\cref{lem:sufficient-trade-offering}, which plays a central role in our proof of the ex-ante guarantee of the Offering Mechanism, provides a sufficient condition for a buyer-seller pair to trade in that mechanism. Here we restate and prove the \lcnamecref{lem:sufficient-trade-offering}.

\begin{lemma}[Restatement of \cref{lem:sufficient-trade-offering}]
	Fix valuation profile $(\Bb,\Bs)$. For every $(i,j)\in M(\Bb,\Bs)$, if $j$ is in $M_{-i}(\Bb,\Bs)$ then buyer $i$ will trade with seller $j$ in the BO  Mechanism, and if $i$ is in $M_{-j}(\Bb,\Bs)$ then buyer $i$ will trade with seller $j$ in the SO Mechanism. Thus, in each such case the edge $(i,j)$ will be traded with probability at least $\nicefrac{1}{2}$ in the Offering Mechanism.
\end{lemma}
\begin{proof}
For every pair $(i,j)\in M(\Bb,\Bs)$, if $j\in M_{-i}(\Bb,\Bs)$, we have by \cref{offering-constraint-vcg} that $\bar{b}=P_i(\Bb,\Bs)\ge s_j$, where the inequality is by \cref{lem:VCG-price}.
Similarly, if $i\in M_{-j}(\Bb,\Bs)$, then $\bar{s}=P_j(\Bb,\Bs)\le b_i$. Therefore, in either case, by \crefpart{ro-trade-probability}{always-trade} the edge $(i,j)$ will be traded with probability at least $\nicefrac{1}{2}$ in the Offering Mechanism.
\end{proof}	

Consider a maximal alternating path that is not a cycle. \cref{lem:even-odd-path} shows that for every seller $j\in A$ that is not at one of the ends such a path, it holds that $j\in M_{-i}$, where $i$ is the buyer that is matched to $j$ in $M$. Here we restate and give the proof of the \lcnamecref{lem:even-odd-path}.

\begin{lemma}[Restatement of \cref{lem:even-odd-path}]
	Let $A$ be a maximal alternating path of $M\cup M_1^*$ that is not a cycle.
	For every seller $j\in A$ who is not at one of the ends of the path,
	it holds that $j\in M_{-i}$, where $i$ is the buyer such that $(i,j)\in M$.
\end{lemma}
\begin{proof}
	By \cref{cor:buyer-first} we can assume w.l.o.g.\ that $A$ starts with a buyer and an edge in $M$.
	So, if the path has an even number of edges, then $A=(i_1j_1i_2j_2...i_{L-1}j_{L-1}i_L)$ and if it is odd then $A=(i_1j_1i_2j_2...i_{L-1}j_{L-1}i_Lj_L)$,
	where in either case each agent $i_l$ denotes a buyer and each agent $j_l$ denotes a seller,
	such that for every $l\in \{1,2,...,L-1\}$ it holds that $(i_l,j_l)\in M$. If the path is odd, it furthermore holds that $(i_L,j_L)\in M$.
	
	We need to show that for every $l\in \{1,2,...,L-1\}$ it holds that  $j_l\in M_{-i_l}$.	
	Assume for contradiction that $j_l\notin M_{-i_l}$ for some $l\in \{1,2,...,L-1\}$.
	Then $M_{-i_l} = M\setminus \{(i_l,j_l)\}$ by subset consistency of the tie-breaking rule, and in $A$ the matching $M_{-i_l}$ matches the set of agents\footnote{We slightly abuse notation by using $A$ to also denote the set of agents in the path $A$.}   $A'= M\cap (A\setminus \{i_l,j_l\})$.

	If the path has an even number of edges, then $i_L\notin A'$.
	To derive a contradiction we observe that the set $A''=A'\cup \{j_l,i_L\}=A\setminus\{i_l\}$
	is matchable (using the edges of $M$ on the path $A$ up to $j_{l-1}$, and the edges of
	$M_1^*$ on the path $A$ starting from $j_l$), and moreover,
	has weight with respect to~$W$ that is at least the weight of $A'$. This holds as $M_1^*$ matched $A''\cap R$ and not $A'\cap R$ for $R = \{j_l,i_{l+1}, j_{l+1},\ldots,i_L\}$, and the weight of $i_L$ is not lower in $W$ than in $W_1$ (and the weight of $j_l$ is the same in both). So we get an contradiction as either $A''$ is matchable and with a higher weight than $A'$ with respect to $W$, or they have the same weight with regard to $W$ and the same weight with regard to $W_1$, and ties were broken differently.
	
	Next we consider the case that the path has an odd  number of edges, in which case also $(i_L, j_L)\in M$.
	It must hold that $s_{j_l}\leq s_{j_L}$ as $M_1^*$ matches $A\setminus \{i_1,j_L\}$ and not the matchable set $A\setminus \{i_1,j_l\}$.
	Recall that since $j_l\notin M_{-i_l}$, then $M_{-i_l} = M\setminus \{(i_l,j_l)\}$, but the matchable set $A'' = A\setminus \{i_l,j_L\}\subseteq A\setminus\{i_l\}$ has at least the weight of
	$A' = A\setminus \{i_l,j_l\}$ with respect to $W$ (since sellers have the same weight in $W$ and $W_1$), so we get an contradiction as either $A''$ is matchable and with a higher weight than $A'$ with respect to $W$, or they have the same weight with regard to $W$ and the same weight with regard to $W_1$, and ties were broken differently in $M_{-i_l}$ and in $M_1^*$.
\end{proof}

\cref{lem:odd-path-extra-edge} considers such paths with an odd number of edges and present some additional characterization that will help us in bounding the GFT of our mechanism.
By Corollary \ref{cor:buyer-first} we can assume w.l.o.g.\ that any  maximal alternating path of  $M\cup M_1^*$ starts with a buyer and an edge in $M$.
Let $GFT_{M'}(U)$ be the GFT of all edges of $M'$ that are contained in $U$. We now restate and prove this \lcnamecref{lem:odd-path-extra-edge}.

\begin{lemma}[Restatement of \cref{lem:odd-path-extra-edge}]
	For $K>3$, let $A=(i_1j_1i_2j_2...i_{L-1}j_{L-1}i_Lj_L)$, be a maximal alternating path of odd number of edges of $M\cup M_1^*$ with
	any agent $i_l$ denoting a buyer and any agent $j_l$ denoting a seller, and the first edge in $M$  ($(i_1,j_1)\in M$).
It holds that
	\begin{itemize}
		\item if $b_{i_L}>b_{i_1}$ then  $i_L\in M_{-j_L}$.
		\item if $b_{i_L}\leq b_{i_1}$ then $GFT_M(A\setminus \{i_L,j_L\} )\geq GFT_{M_1^*}(A) $.
	\end{itemize}
\end{lemma}
\begin{proof}
	We prove that if $b_{i_L}>b_{i_1}$ then  $i_L\in M_{-j_L}$.
	Assume for contradiction that $i_L\notin M_{-j_L}$.
	Since the tie breaking is subset consistent, the matching picked on $A\setminus \{i_L,j_L\}$ will be the same as the one in $M$.
	Yet, the set $A\setminus \{i_1, j_L\}$ is matchable (by the edges of $M_1^*$) and has higher weight than the weight that $M$ gets on $A\setminus \{i_L,j_L\}$, a contradiction.
	
We next consider the case that $b_{i_L}\leq b_{i_1}$.
Let $w=GFT_M(A)=\sum_{l=1}^{L}(b_{i_l}-s_{i_l})$. Notice that:
\[
GFT_M(A\setminus \{i_L,j_L\})=w-(b_{i_L}-s_{i_L})\ge w-(b_{i_1}-s_{i_L})=GFT_{M_1^*}(A\setminus \{i_1,j_L\})=GFT_{M_1^*}(A).\qedhere\]
\end{proof}

We are now ready to complete the proof of \cref{offering-bcfm-pointwise}, which states that for any valuation profile $(\Bb,\Bs)$, the gains from trade of the Offering Mechanism for matching markets is at least half of the from trade of the RVWM mechanism for that profile.

	Fix a valuation profile $(\Bb,\Bs)$.
	To prove the claim we consider the connected components of $M(\Bb,\Bs)\cup M_1^*(\Bb,\Bs)$ and show that in each connected component separately the GFT of the Offering Mechanism in expectation (over the randomness of the mechanism), is at least half the GFT of the RVWM mechanism  of \citet{BrustleCFM17} on $(\Bb,\Bs)$.
	
	By \cref{obs:division-alternating} each connected component is either a maximal alternating path or a cycle.
	By \cref{cor:cycles} any cycle has only two (identical) undirected edges, denote it by $(i,j)\in M(\Bb,\Bs) \cap M_1^*(\Bb,\Bs)$.
	That is, the unique edge $(i,j)$ of $M_1^*= M_1^*(\Bb,\Bs)$ in this cycle is the same as the unique edge $(i,j)$ of $M=M(\Bb,\Bs)$ in that cycle. If $j\in M_{-i}(\Bb,\Bs)$ or $i\in M_{-j}(\Bb,\Bs)$, by \cref{lem:sufficient-trade-offering} buyer $i$ will trade with seller $j$ in the Offering Mechanism with probability at least $\nicefrac{1}{2}$, which obtains at least half the GFT that the RVWM mechanism obtains on $(i,j)\in M_1^*(\Bb,\Bs)$ when the profile is $(\Bb,\Bs)$. Otherwise, since $i\notin M_{-j}$ we have that $\bar{s} \ge b_i$, and since $j\notin M_{-i}$ we have that $\bar{b}\le s_j$.
	Since trade occurs with positive probability on $(i,j)$ in the RVWM mechanism, then by \cref{rvwm-offer}, in this case the GFT of the RVWM mechanism on this edge are therefore at least those of the RO mechanism with SO parameters $\infty$ (no constraint) and $D^B_{b_i}$ (unconditioned distribution) and BO parameters $0$ (no constraint) and~$D^S_{s_j}$ (unconditioned distribution) on $(i,j)$. Since $\bar{s}\ge b_i$ and $\bar{b}\le s_j$, we have by \crefpart{ro-trade-probability}{restrictions-dont-matter} that the probability that trade occurs between $i$ and $j$ is at least as high in our Offering Mechanism (which runs the appropriate RO mechanism, constrained and conditioned) as it is in the unconstrained and unconditioned RO mechanism (that upper-bounds the GFT of RVWM on this edge). Therefore, in this case our Offering Mechanism achieves at least the gains from trade of the RVWM mechanism on this edge (and therefore, on any alternating cycle).
	
	By \cref{cor:buyer-first}, any other maximal alternating path of is not a cycle, and is a path starts or ends with a buyer and an edge from $M$. We will assume w.l.o.g.\ that it start with a buyer and an edge from $M$, and we consider such paths of even and odd numbers of edges separately. Note that Corollary~\ref{cor:buyer-first} implies that there is no connected component that does not include at least one edge from $M$, so by going over all connected components with at least two edges, we cover all the edges of $M_1^*$.
	
	If the number of edges in the path is even, by \cref{lem:even-odd-path}, $M$ matches every seller $j$ in the path to some buyer $i$,
	and for any such pair $(i,j)$ it holds that $j\in M_{-i}$. By \cref{lem:sufficient-trade-offering} buyer $i$ will trade with seller $j$ in the BO Mechanism, so whenever the BO mechanism runs, the	maximal GFT (first best) of that connected component, which is at least the GFT of the $M_1^*$ mechanism for that connected component, will be obtained. The Offering Mechanism runs the BO mechanism is probability $\nicefrac{1}{2}$, so in expectation it obtains at least $\nicefrac{1}{2}$ the GFT of $M_1^*$ for this path.
	
	We next consider the case that the number of edges in the path is odd and at least $3$.\footnote{As noted, if there is a single edge, it is only in $M$. We only need to cover edges in $M_1^*$.}
	Let the path be $A=(i_1j_1i_2j_2...i_{L-1}j_{L-1}i_Lj_L)$ for some $L\geq 2$.
	By \cref{lem:even-odd-path}, for any $l\in \{1,2,...,L-1\}$ it holds that $j_l\in M_{-i_l}$. Next, we use \cref{lem:sufficient-trade-offering} again. We consider two cases, using \cref{lem:odd-path-extra-edge}.
	\begin{itemize}
		\item If $b_{i_L}>b_{i_1}$ then  $i_L\in M_{-j_L}$. In this case all edges of $M$ will each be traded with probability at least $\nicefrac{1}{2}$ in the Offering Mechanisms, so in expectation it obtains at least $\nicefrac{1}{2}$ the GFT of $M$ in this path and thus also at least $\nicefrac{1}{2}$ the GFT of
		$M_1^*$ in this path $A$.
		\item If on the other hand  $b_{i_L}\leq b_{i_1}$ then $GFT_M(A\setminus \{i_L,j_L\} )\geq GFT_{M_1^*}(A) $. Therefore, since every edge $(i_l,j_l)$ for $l\in \{1,2,...,L-1\}$  is traded with probability $\nicefrac{1}{2}$ in the Offering Mechanisms, we have that in expectation the Offering Mechanism obtains in this path $A$ at least~$\nicefrac{1}{2}$ of the GFT of $M_1^*$ in this path $A$.	
	\end{itemize}

We conclude that the Offering Mechanisms obtains at least $\nicefrac{1}{2}$ the GFT that the RVWM mechanism gets on $M_1^*$.
Similar arguments show that the Offering Mechanisms obtains at least $\nicefrac{1}{2}$ the GFT of the RVWM mechanism gets on $M_2^*$.
Thus the Offering Mechanisms, obtains at least $\nicefrac{1}{4}$ the total GFT of $M_1^*$ and $M_2^*$.
The expected GFT of the RVWM mechanism is the average GFT of $M_1^*$ and of $M_2^*$.
We conclude that the Offering Mechanisms obtains at least $\nicefrac{1}{2}$ the GFT of the RVWM mechanism.

\section{The Hybrid Mechanism for\texorpdfstring{\\}{ }Matching Markets: Proofs}
\label{app:matching-hybrid}

In this \lcnamecref{sec:matching-hybrid}, we prove \cref{thm:matching-main}.

\paragraph{Ex-post IR, ex-post (direct trade) weakly budget balance} They directly come from the fact that both the TR mechanism and Offering Mechanism are ex-post IR and ex-post (direct trade) weakly budget balanced.

\paragraph{Bayesian IC} Lemma~\ref{lem:hybrid-BIC} proves that after combining the two mechanisms, the hybrid mechanism is still a BIC mechanism.

\begin{lemma}\label{lem:hybrid-BIC}
The hybrid mechanism for matching markets is BIC.
\end{lemma}

\begin{proof}
We will prove that the Hybrid Mechanism is BIC for the seller. A similar argument holds for the buyer.
For each seller $j$ with cost $s_j$, suppose she misreports her cost to be $s_j'\not=s_j$. We will show that taking expectation over other agents' valuation profile $\Bb,s_{-j}$, the expected utility of $s_j$ when reporting truthfully is at least the expected utility of seller $j$ with true cost $s_j$ when reporting $s_j'$.
We consider three cases:
\begin{itemize}
\item
First, consider the case where when seller~$j$ reports $s_j$, then she is in the first-best and the TR Mechanism is run. In this case, we note that if seller~$j$ misreports to $s'_j$, then either the first-best is unchanged (and so the TR is still run) or $j$ is no longer in the first-best. In the former case, seller~$j$ does not profit since by \cref{thm:matching-TR} the TR Mechanism is ex-post IC, and in the latter case seller~$j$ does not profit as she gets utility $0$.
\item
Now, consider the case where when seller~$j$ reports $s_j$, then she is in the first-best and the Offering Mechanism is run with an offer on the edge $(i,j)$. Similarly to above, we note that if $j$ misreports to $s'_j$, then either the first-best is unchanged (and so the Offering Mechanism is still run) or $j$ is no longer in the first-best, and has $0$ utility. In particular, $j$ cannot cause the TR mechanism to run without getting $0$ utility. Also note that for the same reason, for every report of buyer $i$ that keeps $(i,j)$ in the first best, the Offering Mechanism is still run, and so it is enough to show truthfulness of $j$ in expectation over all such reports of buyer~$i$, and we have shown precisely that in \cref{offering-bic} using \crefpart{ro}{properties}. So, we have that when $j$ has cost $s_j$ such that there exists $\Bb,s_{-j}$ such that the Offering Mechanism is run with an offer on the edge $(i,j)$, then in expectation over all such $\Bb,s_{-j}$, it is the case that $s_j$ cannot gain from misreporting.
\item
Finally, consider the case where when seller~$j$ reports $s_j$, then she is not in the first-best. In this case, regardless of the mechanism that is actually run when $j$ reports $s_j$, her outcome reporting $s_j$ would have been the same under both mechanisms. So, since we have shown that when $j$ is not in the first-best, truthtelling is ex-post IC in both mechanisms, we have that this implies that truthfulness is ex-post IC for seller~$j$ in the hybrid mechanism in this case.\qedhere
\end{itemize}
\end{proof}

\paragraph{Ex-post efficiency guarantee} Whenever $\alpha(\Bb,\Bs)\geq 1/2$, the hybrid mechanism run TR mechanism. The ex-post guarantee directly comes from Claim~\ref{matching-TR-gains}.

\paragraph{Ex-ante efficiency guarantee}
Let $(\Bb,\Bs)$ be a profile. If $\alpha(\Bb,\Bs)\geq \frac{1}{2}$, the hybrid mechanism runs the Trading Reduction mechanism, which achieves at least $\nicefrac{1}{2}$-fraction of the first-best gains from trade. This is at least $\nicefrac{1}{2}$-fraction the gains from trade of the RVWM mechanism. If $\alpha(\Bb,\Bs)<\frac{1}{2}$, the hybrid mechanism runs the Offering Mechanism, which by \cref{offering-bcfm-pointwise} achieves at least a $\nicefrac{1}{2}$-fraction of the GFT of the RVWM mechanism for this profile. So, for any profile the hybrid mechanism achieves at least a $\nicefrac{1}{2}$-fraction of the GFT of the RVWM mechanism for this profile, and so by \cref{bcfm}, it achieves at least a $\nicefrac{1}{4}$-fraction of the GFT of the second-best mechanism, as required.

\begin{proof}[Proof of \cref{cor:matching-main-beta}]
Since $\beta(\Bb,\Bs)\le\alpha(\Bb,\Bs)$ (see the proof of \cref{cor:matching-TR-beta}), we have that in this case also $\alpha(\Bb,\Bs)\ge\nicefrac{1}{2}$, and so the hybrid mechanism runs the TR mechanism for matching markets, and so the claim following via \cref{cor:matching-TR-beta}.
\end{proof}

\end{document}